\documentclass{sig-alternate}

\usepackage{amsmath}
\usepackage{algorithm2e}
\usepackage{algorithmic}
\usepackage{amsfonts}
\usepackage{appendix}
\newtheorem{thm}{Proposition}
\newtheorem{lma}{Lemma}
\newdef{defn}{Definition}
\newtheorem{cor}{Corollary}
\usepackage{graphicx,color}
\usepackage{caption}
\DeclareCaptionType{copyrightbox}
\usepackage{subcaption}
\usepackage{multirow}
\begin{document}

\title{Mathematical Frameworks for Pricing in the Cloud: Revenue, Fairness, and Resource Allocations}

\numberofauthors{1}
\def\sharedaffiliation{%
\end{tabular}
\begin{tabular}{c}}

\author{
\alignauthor Carlee Joe-Wong \hspace{0.5in} Soumya Sen \\
\sharedaffiliation
\affaddr{Princeton University} \\
\affaddr{\{cjoe,soumyas\}@princeton.edu}
}

%\numberofauthors{3}
%
%\author{
%\alignauthor Some \\
%\affaddr{who}\\
%\email{@}
%%
%\alignauthor anonymous \\
%\affaddr{work}\\
%\email{an}
%%
%\alignauthor people \\
%\affaddr{together}\\
%\email{institution}
%}

\maketitle

\begin{abstract}
As more and more users begin to use the cloud for their computing needs, datacenter operators are increasingly pressed to effectively allocate their resources among these client users. Yet while much work has been done in this area, relatively little attention has been paid to studying perhaps the ultimate lever of resource allocation: pricing. Most data centers today charge users by ``bundling'' heterogeneous resources together in a fixed ratio and selling these bundles to their clients. But bundling masks the fact that different users require different combinations of resources (e.g., CPUs, memory, bandwidth) to process their jobs. The presence of multiple resources in fact allows an operator to offer many different types of pricing strategies, which may have different effects on its revenue. Moreover, to avoid user dissatisfaction, operators must consider the impact of their chosen prices on the fairness of the jobs processed for different users. In this paper, we develop an analytical framework that accounts for the fairness and revenue tradeoffs that arise in a datacenter's multi-resource setting and the impact that different pricing plans can have on this tradeoff. We characterize the implications of different pricing plans on various fairness metrics and derive analytical limits on the operator's fairness-revenue tradeoff. We then provide an algorithm to navigate this tradeoff and compare the tradeoff points for different pricing strategies on a data trace taken from a Google cluster.
%in a  within an analytical frameworks that capture all these aspects of cloud pricing are currentl   
%require developing an analytical framework that accounts for the fairness and efficiency tradeoffs that arise in such multi-resource allocation settings, and the impact that pricing schemes can have on this tradeoff.  In this work, 
% We investigate the pricing of bundled and unbundled resources, and show that bundling can lead to a significant loss in revenue. Moreover, we construct a mathematical framework for optimizing fairness and show that this framework recovers previously proposed metrics for the fairness of multi-resource allocation. We then explore the tradeoff between fair prices and those yielding higher revenue for the operator. Finally, we apply results from user surveys on multi-resource fairness to suggest client preferences for different pricing schemes.
\end{abstract}

\vspace{-0.1in}
\vspace{-0.1in}

\section{Introduction}

% Introduce the problem.

\subsection{Cloud Computing Resources}

As Internet connectivity becomes ever more ubiquitous, cloud computing is becoming an increasingly popular alternative to traditional computing paradigms \cite{weinhardt2009cloud}. This dramatic increase in demand has led to several research challenges: how should a datacenter allocate multiple shared resources among users in a practical, scalable manner? The research implications of this question range from developing an efficient method of routing client requests to different datacenters, to allocating a certain amount of resources for processing each job as it arrives at a particular datacenter. In this work, we focus on an ``infrastructure-as-a-service'' view of the cloud, in which the resources considered are raw computing resources such as CPU time or memory. % Perhaps the most pressing question, however, deals with the management of fluctuating user demand for cloud resources \cite{zhang2010cloud}). As users submit computational jobs to a datacenter, the operator must decide how much of the datacenter resources to allocate to this particular user's job, while taking into account the needs of other users and the resource availability.

% The exact definition of ``cloud computing'' is a matter of some debate, as datacenter operators can offer several different service levels running in the cloud. At its most sophisticated, datacenters can provide entire software platforms, while at its most basic, the datacenter can simply provide computational and storage resources (e.g., CPU time and RAM memory) on which users can run computing jobs. In this work, we consider the latter definition of cloud computing: a datacenter operator provides users with computing resources for a certain amount of time, which are used to run computational jobs. Such a view of cloud computing is sometimes termed ``infrastructure-as-a-service,'' as opposed to ``platform-as-a-service'' or ``software-as-a-service.'' 

When provisioning resources to users, datacenter operators must take into account the fact that multiple types of resources are being allocated. For instance, users' jobs require both CPU time and RAM, which are non-substitutable resources: more RAM does not mitigate a need for CPU processing time, and neither does more processing time significantly lessen RAM requirements. This provisioning of multiple resources thus introduces new challenges to the datacenter operator's resource allocation problem, which are only beginning to be studied \cite{joe2012multi}. In this work, we consider the consequences of multiple resource allocation on the prices offered to datacenter users.

\subsection{Pricing Objectives}\label{sec:challenges}

As datacenters become ever more popular, many position papers have argued for the need to research new pricing mechanisms that can accommodate the competing objectives of a datacenter operator \cite{durkee2010cloud,weinhardt2009cloud}. While the operator of course wishes to maximize its revenue, its pricing policy must also take into account the dynamic nature of user demand: as users submit jobs to the datacenter, the operator must be able to accommodate the influx of new jobs and balance the needs of incoming users with those of users already in queue. We thus introduce a fairness dimension: to avoid incurring user dissatisfaction, the datacenter operator should allocate its resources so that each user can process a fair, or equitable, number of jobs. Treating users fairly, however, can in itself impact revenue. In this work, we introduce an optimization framework that allows datacenter operators to balance the revenue and fairness of its resource allocation through setting prices for its computing resources.

The presence of multiple resources gives a datacenter operator flexibility in offering different types of pricing plans. In this work, we examine the implications of different pricing schemes on the achieved fairness and revenue. We consider three different pricing plans:

{\bf Bundled pricing:} Bundled pricing assumes that operators group heterogeneous resources such as CPU time and memory together and then sell this resource bundle to users. For instance, Google's Compute Engine sells virtual machines (VMs), each with a fixed CPU and memory capacity \cite{google}. Basic pricing policies for Amazon's EC2 similarly offer VM bundles \cite{ec2}. Users are charged at a fixed unit rate per VM; they are free to make use of all or part of each VM resource. For instance, a user with computationally intensive jobs might not require the VM's full RAM capacity, but will need to purchase its use anyway.

{\bf Resource pricing:} Under resource pricing, an operator charges users separate unit prices for each resource that they use. Microsoft's Windows Azure, for instance, charges separate unit rates for CPU time and memory \cite{azure}. Unlike bundled pricing, resource pricing does not charge users for resources they do not use. One can, however, view resource pricing as a form of bundled pricing, in which the bundles for each user are customized to exactly fit that user's resource requirements. The price of each bundle is determined by the unit prices of the different resources.

{\bf Differentiated pricing:} In differentiated pricing, users are allocated exactly as much of each resource as they require to process their jobs. Differentiated pricing differs from resource pricing, however, in that the operator can freely determine each user's per-job cost, and is not constrained to charge users by unit resource prices. While differentiated pricing has not yet been introduced to the market, many believe that it will be necessary in the future \cite{odlyzko2009network}.

In addition to the three pricing plans introduced above, one may take advantage of datacenters' economies of scale by introducing {\bf volume discounts} into a pricing policy; Amazon's EC2 offers such a pricing plan, charging a lower unit rate for users who purchase a large quantity of resources \cite{ec2}. While some works have investigated the effects of volume discounts, e.g., for contract users \cite{du2012pricing}, none have done so in a multi-resource context. In this work, we compare bundled, resource, and differentiated pricing with volume discounts in a mathematical framework that takes into account user demand, operator revenue, and fairness across users.

While the pricing plans offered by specific operators have been analyzed and compared \cite{siham2012price,wang2010distributed}, we use a general model to solve for an operator's optimal prices. Our work is unique in accounting for both the pricing of multiple resources and the fairness of the resulting allocation. While other papers have accounted for fairness by using game theory to model user competition \cite{teng2010resource} or studying the strategy-proof property \cite{teo2009strategy}, these works do not explicitly deal with the presence of multiple, non-substitutable resources. Works that do consider multi-resource fairness, on the other hand, do not include user pricing. Instead, operators are assumed to allocate a certain number of jobs to users in a manner that is both efficient, in terms of a large number of processed jobs, and equitable, i.e., no user processes a disproportionately high or low number of jobs \cite{ghodsi2011dominant,joe2012multi,wagner2012autonomous}.

\subsection{Job Arrivals and Deadlines}

Dynamic pricing has been perhaps the most common approach taken in the literature on pricing for cloud or datacenter computing. Some papers have investigated revenue-maximizing prices \cite{feng2012revenue,xu2012maximizing}, while other models take into account such factors as job completion deadlines \cite{wang2012datacenter}. In another variant, users' willingness to pay and job prices vary depending on a given job's completion time \cite{tsakalozos2011flexible}. For instance, operators may give different price quotes to users as a function of the completion time, and the user can select from these quotes \cite{henzinger2010flexprice}. Such a formulation may be used to maximize the long-term social surplus, defined in terms of user utilities over time \cite{menache2011socially}. However, most of these works do not take the presence of multiple resources into account. Other works have proposed dynamic auctions, an approach popular with grid computing \cite{greenberg2008cost,samimi2011review}. However, such measures, e.g., the spot pricing offered by Amazon's EC2, have been shown to be ineffective in practice \cite{wee2011debunking}.

In our work, we suppose that the datacenter operator optimizes its prices over a finite time horizon. We divide time into discrete intervals, e.g., of one hour, and suppose that users are allocated the use of a given resource (e.g., CPU) for each particular interval.\footnote{If the time intervals are very short, this pricing is sometimes called \emph{resource-as-a-service}, i.e., the selling of individual resources at very short time granularities \cite{ben2012resource}.} We then have two possibilities: in the first, user demands could be the same during each time interval. The operator can then compute its prices in an offline optimization and simply offer those prices during each interval. In the second, more realistic, scenario, users' demands change over time, and the operator must change its prices accordingly. Operators must then account for job deadlines, so that users' jobs complete in a time acceptable to users. In the main body of this work, we consider the first scenario of optimization in a single time interval; Appendix \ref{sec:deadlines} shows how our work may be extended to account for changes in user demand and job deadlines.

We begin developing our framework in Section \ref{sec:user} by examining users' choices of how many jobs to submit, given a set of offered prices. Section \ref{sec:pricing} then discusses bundled, resource, and differentiated pricing in the context of users' choices and derives expressions for the operator's revenue. In Section \ref{sec:fair_rev}, we introduce a family of metrics to measure the fairness under different pricing plans. We then characterize the fairness-revenue tradeoff with these metrics and derive algorithms for practically choosing prices at different points on the tradeoff. Finally, in Section \ref{sec:numerical} we perform simulations on a dataset from a Google cluster to compare the achieved fairness and efficiency for different pricing plans and a range of resource capacities and volume discounts. Section \ref{sec:conclude} concludes the paper.

\section{User Optimization}\label{sec:user}

Suppose that a datacenter operator has multiple clients (i.e., users), each of whom submits jobs to the datacenter. Each job requires a certain amount of datacenter resources, e.g., CPU time, memory, and bandwidth; we suppose, as is often the case in practice, that these resources are not substitutable. We index the resources with the variable $i$, and denote the amount of resource $i$ required for each job by $R_{ij}$, where $j$ indexes the user type.  Different types of users are distinguished by the different resource requirements of their jobs; in an extreme case, each user might submit jobs with slightly different resource requirements and thus belong to a distinct type. Multiple users could also be grouped together into the same type, if their jobs have similar resource requirements that we can approximate as being the same. We fix the number of user types as $n$ and the number of resources as $m$. If users are grouped together into a few types, then we could have $n \leq m$; otherwise, in general we have $n \geq m$ as there are many more user clients than types of resources.

When deciding how many jobs to submit, users take into account the prices set by the datacenter operator, which yield a cost per job for each user $j$.  The structure of this pricing plan may vary; Section \ref{sec:pricing} introduces three different ways (bundled, resource, and differentiated pricing) in which the prices may be set. In all three pricing plans, each user has a fixed per-job cost, which we denote as $r_j$ for user $j$. This cost accounts for both the resource prices and the volume discounts, which are taken to be separate functions of the amount of each resource required by the user. We parameterize the volume discount with a constant $\gamma\in(0, 1]$; if a user requires $R_{ij}x_j$ amount of resource $i$, where $x_j$ denotes the number of jobs submitted by user $j$, then the user is charged according to the discounted quantity $R_{ij}^\gamma x_j^\gamma$. Thus, a larger $\gamma$ corresponds to a smaller volume discount, and $\gamma = 1$ corresponds to no volume discount. Since the amount of resources required is an affine function of the user's demand, we let $r_j$ denote the per-job cost with the volume discount and separately account for the discounted $x_j^\gamma$ term.

Users derive utility from the number of jobs that the datacenter processes, but lose utility from the amount that their jobs cost.\footnote{We assume throughout that the operator sets prices so that all submitted jobs can be processed with the resources available. ``Submitted'' and ``processed'' can thus be used interchangeably.} We use $x_j$ to denote the number of jobs processed for each user $j$, and use the function $U_j(x_j)$ to denote the utility from these jobs. Since a user's utility increases with the number of jobs processed, we follow the economic principle of diminishing marginal utility and assume that $U_j$ is a strictly increasing, concave function, expressed in units of dollars. Subtracting the cost $r_j x_j^\gamma$ from the utility of processing $x_j$ jobs, we find that each user $j$ maximizes
\begin{equation}
U_j\left(x_j\right) -r_j x_j^\gamma.
\label{eq:user_max}
\end{equation}
For simplicity, we assume that $U_j$ is continuously differentiable for $x_j > 0$. We suppose that the users have sufficient available cash to pay for any number of jobs $x_j$, so long as (\ref{eq:user_max}) is positive; for instance, enterprise users would meet this requirement. Users thus have no explicit budget constraint; they submit as many jobs as required to maximize (\ref{eq:user_max}) for a given set of prices. % subtract the cost of a job from the utility function $U_j$ rather than impose an explicit budget constraint (i.e., maximum budget) due to the fact that a user would likely derive greater utility from the same number of jobs as the price of resources decreased. Moreover, the resulting savings could then be used to send more jobs to the datacenter, relieving the user of the burden of processing all jobs with her own resources.

An example family of utility functions are the so-called \emph{isoelastic} functions
\begin{equation}
U_j(x) =
\begin{cases}
c_j(1 - \alpha_j)^{-1}x^{1 - \alpha_j}, &\alpha_j\in(0, 1) \\
c_j \log(x) &\alpha_j = 1
\end{cases}
\label{eq:alphafair}
\end{equation}
commonly used in economics \cite{textbook}. In (\ref{eq:alphafair}), $\alpha_j$ parameterizes the concavity of the utility function (i.e., quantifying the degree of diminishing marginal utility) and $c_j > 0$ is a positive constant that scales the utility level for user $j$. As $\alpha_j$ increases, $U_j$ becomes more concave, and the user becomes more price-sensitive since the cost term $r_j{x_j^\star}^\gamma$ term becomes more dominant. % We emphasize that the ``$\alpha$-fairness'' here should not be construed as enforcing any sort of fairness across different users; we use this term because it is a conventional name for functions of this type.\footnote{In Section \ref{sec:fair_rev}, we use a variant of the $\alpha$-fair functions (\ref{eq:alphafair}) to measure the fairness of a particular distribution of jobs among users.}

We assume that the user receives zero utility if no jobs are processed; this assumption may be enforced by taking $U_j(0) = 0$. The maximizer $x_j^\star$ of (\ref{eq:user_max}) then satisfies
\begin{equation}
U_j'\left(x_j^\star\right) = r_j\gamma{x_j^\star}^{\gamma - 1},
\label{eq:demand_cond}
\end{equation}
where $x_j^\star$, a function of the per-job price $r_j$, denotes the utility-maximizing number of jobs. % For instance, if $U_j$ is an $\alpha$-fair utility function (\ref{eq:alphafair}) and $d(x_j) = x_j$ (no volume discount), then (\ref{eq:demand_cond}) is satisfied by $x_j^\star = \left(r_j/c_j\right)^{1/(1 - \alpha_j)}$. In general, if there is no volume discount ($d(x) = x$), then $x_j^\star = {U_j'}^{-1} \left(r_j\right)$.

We assume that the operator chooses a value of $\gamma$ such that a solution to (\ref{eq:demand_cond}) exists, and that (\ref{eq:user_max}) has non-positive second derivative at this value of $x_j^\star$, ensuring that the solution is indeed a maximum. Mathematically, the second derivative condition may be expressed as
\begin{equation}
U_j''\left(x_j^\star\right) < \gamma(\gamma - 1){x_j^\star}^{\gamma - 2}r_j.
\label{eq:second_der}
\end{equation}
For instance, the isoelastic utility functions satisfy this condition if $\gamma > 1 - \alpha_j$. We can then show that user $j$'s demand $x_j^\star$ is decreasing in $r_j$. By differentiating (\ref{eq:demand_cond}), we find that this rate of decrease is
\begin{equation}
{x_j^\star}'(r_j) = \frac{\gamma {x_j^\star}^{\gamma - 1}}{U_j''\left(x_j^\star\right) + \gamma(1 - \gamma){x_j^\star}^{\gamma - 2}r_j},
\label{eq:xder}
\end{equation}
which is negative if (\ref{eq:second_der}) holds.

\section{Pricing Datacenter Jobs}\label{sec:pricing}

The presence of multiple datacenter resources allows the datacenter operator to offer many different pricing plans, each of which may be combined with volume discounts. In Section \ref{sec:plans}, we discuss the three pricing plans introduced in Section \ref{sec:challenges}: bundled, resource, and differentiated pricing. We incorporate these prices in the user-optimization framework introduced in the previous section and then compare the optimal prices and resulting revenue in Section \ref{sec:compare}.

% {\color{red} need to add examples? Or put that in the intro?}

\subsection{Pricing Strategies}\label{sec:plans}

We first consider {\bf bundled pricing}, in which the operator groups different resources in a fixed ratio--for instance, 1 CPU and 2 GB of RAM might be grouped together as a bundled resource. The operator can then charge users a unit price per bundle of resources required for the user's jobs. We let $b_i$ denote the amount of each resource $i = 1,\ldots,m$ in each bundle, and let $p$ denote the per-bundle price. Using the notation from Section \ref{sec:user} to describe users' utility functions and per-job resource requirements, we find that each user $j$ requires
\begin{equation}
\mu_j = \max_i\left(\frac{R_{ij}}{b_i}\right)
\label{eq:domshare}
\end{equation}
bundles in order to complete one job. To calculate the per-job cost, we incorporate the volume discount introduced in the previous section, parameterized by $\gamma\in(0,1]$. The per-job cost then becomes $r_j = \mu_j^\gamma p$. Moreover, the number of bundles available is $\min_i C_i/b_i$, where $C_i$ denotes the capacity of resource $i$. Thus, the operator faces the resource constraint
\begin{equation}
\sum_{j = 1}^n \max_i\left(\frac{R_{ij}}{b_i}\right)x_j^\star\left(\mu_j^\gamma p\right) \leq \min_i\frac{C_i}{b_i},
\label{eq:bundlecon}
\end{equation}
where $x_j^\star\left(\mu_j^\gamma p\right)$ denotes the user's demand for jobs at price $p$, given the bundle $\left(b_1,\ldots,b_m\right)$. The resulting revenue under bundled pricing, $\rho_b$, may be expressed as
\begin{equation}
\rho_b = p\sum_{j = 1}^n \left(\mu_jx_j^\star\left(\mu_j^\gamma p\right)\right)^\gamma.
\label{eq:bundlerev}
\end{equation}

We next observe that $\rho_b$ in (\ref{eq:bundlerev}) is a non-differentiable function of the amount of each resource $i$ (i.e., $b_i$) in the bundle; thus, the operator will likely find it difficult to optimize over the $b_i$. Instead, it can fix these $b_i$ and simply optimize over the unit price $p$. For instance, the $b_i$ might be chosen as proportional to the resource capacities $C_i$, i.e., such that for some constant $b$, $b = C_i/b_i$ for each resource $i$. In this case, the user's bundle requirement (\ref{eq:domshare}) is proportional to her dominant share, which is defined as follows:
\begin{defn}
A user's {\bf dominant share} is the maximum share of any resource allocated to the user. Thus, if a user receives $R_{ij}x_j^\star$ amount of each resource $i$, her dominant share would be $\max_i R_{ij}x_j^\star/b_i$, i.e., the bundle requirement (\ref{eq:domshare}) multiplied by $b$.
\end{defn}
In Section \ref{sec:fairness}, we show that users' dominant shares may be used to define the fairness of their utilities received.

While bundled pricing is convenient for the operator in the sense that only one price variable needs to be chosen, bundled pricing may lead to wasted capacity. Since different users' jobs have different resource requirements, the ratio of resources in each bundle will not in general match the ratio of per-job resource requirements for a user. Thus, some users may be forced to purchase excess capacity of some resources--for instance, a user with a per-job requirement of 2 CPU cores and 1 GB of RAM would need to purchase 4 CPU cores and 1 GB of RAM for each job if the operator offered bundles of 4 CPU cores and 1 GB of RAM. To avoid this wasted capacity, which could potentially be sold to another user with different resource requirements, an operator could offer {\bf resource pricing}, in which the unit price of each resource is independently determined by the operator, and users are permitted to purchase exactly as much of each resource as they need to process their jobs.

In this pricing framework, we let $p_i$ denote the unit price of each resource $i$, so that each user $j$'s per-job cost is $r_j = \sum_{i = 1}^m p_i R_{ij}^\gamma$. Thus, the operator's total revenue $\rho_r$ is
\begin{equation}
\rho_r = \sum_{i = 1}^m p_i \sum_{j = 1}^n \left(R_{ij}x_j^\star\left(\sum_{i = 1}^m p_iR_{ij}\right)\right)^\gamma.
\label{eq:revenue}
\end{equation}
The resource constraints, which require that the operator be able to provide enough resources to each user, can then be expressed as
\begin{equation}
\sum_{j = 1}^n R_{ij}x_j^\star \leq C_i
\label{eq:constr}
\end{equation}
for each resource $i$. The left-hand side corresponds to the sum of users' required amounts of resource $i$, while the right side is simply the capacity of resource $i$.

We next note that resource pricing may be viewed as a form of bundled pricing in which the ratio of resources in each bundle exactly matches that required by users. This relationship between bundled and resource pricing suggests that the operator can offer a third, even more customized type of pricing--{\bf differentiated pricing}, in which the operator not only allows users to receive only as much of each resource as they need, but independently sets the prices of these customized bundles. Thus, the operator can set a different per-bundle price $\overline{p}_j$ for each user $j$. We suppose that one bundle exactly equals the user's resource requirements for one job, so that the per-job cost $r_j = \overline{p}_j$. The operator's revenue $\rho_d$ under differentiated pricing is then
\begin{equation}
\rho_d = \sum_{j = 1}^n \overline{p}_jx_j^\star\left(\overline{p}_j\right)^\gamma,
\label{eq:differentialrev}
\end{equation}
and the resource constraints are exactly those under resource pricing (\ref{eq:constr}).

\subsection{Comparing Pricing Plans}\label{sec:compare}

In this section, we compare the operator's revenue under the three pricing plans introduced in the previous section. We suppose that for any of the three pricing plans, the operator chooses the prices so as to optimize a given objective function of users' per-job costs and demands, e.g., the revenue (\ref{eq:bundlerev},\ref{eq:revenue},\ref{eq:differentialrev}). In Section \ref{sec:fairness}, we define the fairness of users' utilities, which form another family of possible objective functions.

We first give sufficient conditions under which the objective value under bundled pricing is at most that under resource pricing:
\begin{lma}\label{lem:bundle}
If $R_{ij}/b_i$ for each user $j$ is maximized by resource $i = k$, then the maximum objective function value with bundling does not exceed the maximum objective with unbundled resources.
\end{lma}
\begin{proof}
Suppose without loss of generality that $k = 1$. Each user's per-job bundle requirement is then $R_{1j}/b_1$. Let the optimal price with bundling be denoted by $p^b$. Then if the operator charges $p_1 = p^b/b_1$ for resource 1 and nothing for any other resource, each user's per-job cost without bundling is $R_{1j}p^b/b_1$, which equals the cost with bundling. Each user's demand is solely determined by the price per job, so each user demands the same amount of jobs with and without bundling, and the resource constraints are satisfied. Since the objective function can be written entirely in terms of users' per-job costs, the value of the objective function is also unchanged. Thus, the maximum objective without bundling is at least as much as that with bundling.
\end{proof}
Given a set of resource requirements, a bundle always exists that satisfies the requirements of Lemma \ref{lem:bundle}--for instance, the operator might simply set $b_1$ to be extremely small. Such a scenario, however, is somewhat artificial. If an operator instead sets the ratio of resources in each bundle equal to the ratio of resource capacities, then Lemma \ref{lem:bundle} holds if each user has the same dominant resource, defined as follows:
\begin{defn}
A given user's {\bf dominant resource} is the resource $k$ which maximizes $R_{ij}/C_i$ over all resources $i$.
\end{defn}
We thus have the following corollary:
\begin{cor}\label{cor:domresource}
Suppose that the ratio of resources in each bundle is equal to the ratio of resource capacities. Then if each user has the same dominant resource, the operator's maximum objective function value with bundling does not exceed that without bundling.
\end{cor}
\begin{proof}
If the amount of each resource $i$ in the bundle is $b_i$, then there exists a constant $B$ such that $Bb_i = C_i$ for each resource $i$. Thus, the resource $k$ maximizing $R_{ij}/C_i$ also maximizes $R_{ij}/b_i$, and Lemma \ref{lem:bundle} applies.
\end{proof}
Intuitively, one might expect a result of this type, as the operator can choose only one price $p$ under bundled pricing, while it chooses $m \geq 1$ resource prices $p_i$ under resource pricing. We note, however, that in general the operator's revenue under bundled pricing may not be less than that under resource pricing--under bundled pricing, users often are required to purchase excess capacity, and the resulting excess revenue may offset the operator's decreased flexibility in choosing prices. However, when comparing resource pricing and differentiated pricing, we obtain a stronger result:
\begin{thm}\label{prop:resource_diff}
If an operator prices jobs by unit prices $p_i$ for each resource $i$, its maximum objective function value does not exceed the optimal value under differentiated pricing. If the matrix of resource requirements ${\bf R}$ has rank $n$, then the optimal objective function value under differentiated and resource pricing is the same.
\end{thm}
\begin{proof}
Suppose that the optimal prices under resource pricing are $p_i$ for each resource $i$. Let the differentiated prices $\overline{p}_j = \sum_{i = 1}^m R_{ij}^\gamma p_i$, so that the price per job does not change for any user. Then $x_j^\star$, the user's demand for jobs, does not change either, as it depends solely on the per-job cost. The objective function value under these differentiated prices thus attains the maximum revenue under resource pricing. Moreover, since the resource constraints depend only on the number of jobs $x_j^\star$, these differentiated prices are a feasible solution.

If the matrix $R$ has rank $n$, then the resource prices may be set in such a way that the per-job cost for each user equals that under differentiated pricing. Thus, since the objective depends only on per-job costs, the optimal objective value under resource and differentiated pricing is the same.
\end{proof}
In general, a datacenter operator will have several more clients than resources (i.e., $m \leq n$)--the number of datacenter resources will generally be a small number, e.g., two for CPU cores and memory. However, if we instead group users into a few distinct types, then the $m\times n$ matrix of resource requirements ${\bf R}$ may have rank $n \leq m$.

\section{Fairness and Revenue}\label{sec:fair_rev}

In this section, we examine a datacenter operator's two objectives: revenue and fairness. We first characterize the operator's revenue in Section \ref{sec:revenue}, and then consider fairness in Section \ref{sec:fairness}. Section \ref{sec:tradeoff} then constructs an optimization framework that takes into account both revenue and fairness.  We derive bounds on the tradeoff between fairness and revenue and give an algorithm that finds the optimal prices for bundled, resource, and differentiated pricing at various points on this tradeoff curve.

\subsection{Revenue and Resource Constraints}\label{sec:revenue}

From the previous section, we see that the operator's revenue may be expressed as
\begin{equation}
\rho =
\begin{cases}
p\sum_{j = 1}^n \left(\mu_jx_j^\star\left(\mu_j^\gamma p\right)\right)^\gamma \\
\sum_{i = 1}^m p_i \sum_{j = 1}^n \left(R_{ij}x_j^\star\left(\sum_{i = 1}^m p_iR_{ij}\right)\right)^\gamma \\
\sum_{j = 1}^n \overline{p}_jx_j^\star\left(\overline{p}_j\right)^\gamma \\
\end{cases}
\label{eq:rev_all}
\end{equation}
for the bundled, resource, and differentiated pricing plans respectively. Should the operator wish to maximize its revenue, a natural approach would be to take the first derivative of the revenue with respect to each price variable and set it equal to zero. In this case, however, we can show that the revenue is always decreasing in each price variable, assuming that the utility function $U_j$ is not too concave compared to the cost $r_j {x_j^\star}^\gamma$. Intuitively, each user $j$'s total utility should grow slowly enough with the prices so that the optimal demand $x_j^\star$ is well-defined, yet grows sublinearly so that the overall amount paid $r_j{x_j^\star}^\gamma$ decreases with the prices ${|bf p}$.
\begin{thm}\label{prop:decrevbundle}
Revenue $\rho$ is decreasing in price if, for all users $j$,
\begin{equation}
0 < U_j''\left(x_j^\star\right) + \gamma r_j{ x_j^\star}^{\gamma - 2} < \gamma^2 r_j{x_j^\star}^{\gamma - 2},
\label{eq:cond_dec_rev}
\end{equation}
where $r_j$ denotes the per-job cost for user $j$ and is an affine function of the prices offered. For instance, the isoelastic utility functions (\ref{eq:alphafair}) satisfy these criteria if, as assumed in Section \ref{sec:user}, $\gamma > 1 - \alpha_j$.
\end{thm}
\begin{proof}
We take the first derivative of the revenue $\rho$ with respect to a price variable $p_k$ (which may correspond to any of the three pricing plans) to obtain
\begin{equation*}
\frac{\partial\rho}{\partial p_k} = \sum_{j = 1}^n \frac{\partial r_j}{\partial p_k} x_j^\star\left(r_j\right)^\gamma + \sum_{j = 1}^n r_j\frac{\partial r_j}{\partial p_k}\gamma{x_j^\star}^{\gamma - 1}{x_j^\star}'\left(r_j\right).
\end{equation*}
Since ${x_j^\star}' < 0$, this derivative is negative if, for each user $j$,
\begin{equation*}
-{x_j^\star}' > \frac{{x_j^\star}}{\gamma r_j}.
\end{equation*}
We can then use (\ref{eq:xder}) to obtain the equivalent condition
\begin{equation*}
U_j'' > -\gamma r_j {x_j^\star}^{\gamma - 2},
\end{equation*}
which is the first inequality in (\ref{eq:cond_dec_rev}). The second inequality follows from our assumption of concavity (\ref{eq:second_der}) in Section \ref{sec:user}.

If we assume isoelastic utility functions for each user, as in (\ref{eq:alphafair}), then (\ref{eq:cond_dec_rev}) becomes
\begin{equation*}
0 < -c_j\alpha_j{x_j^\star}^{-\alpha_j - 1} + \gamma r_j{x_j^\star}^{\gamma - 2} < \gamma^2 r_j {x_j^\star}^{\gamma - 2},
\end{equation*}
where $x_j^\star = \left(r_j/c_j\right)^{1/(1 - \alpha_j - \gamma)}$. Thus, we find the condition
\begin{equation*}
0 < -\alpha_j\gamma + \gamma < \gamma^2,
\end{equation*}
or equivalently $0 < 1 - \alpha_j < \gamma$, which we have assumed to hold in Section \ref{sec:user}.
\end{proof}
Thus, the operator can maximize its revenue by offering the lowest possible prices compatible with its resource constraints, i.e., such that the resources required to process all users' jobs do not exceed the resource capacities. We next show that if users have isoelastic utility functions (\ref{eq:alphafair}), then the amount of each resource required to process users' demanded jobs is a decreasing function of the prices. The resource constraints thus impose a lower bound on the prices offered; in fact, they define a convex set of the prices:
\begin{lma}\label{lem:resourceconvex}
The set $\left\{{\bf p}\;|\; {\bf R_ix^\star}({\bf r}({\bf p})) \leq C_i,\;i = 1,2,\ldots,n\right\}$ is convex, where ${\bf r}$ is a vector of the per-job price for each user and ${\bf p}$ a vector of the prices chosen by the operator. Moreover, ${\bf R_ix^\star}({\bf r}({\bf p}))$ is decreasing in each price $p_k$.
\end{lma}
\begin{proof}
We consider only resource pricing; the result may be similarly proved for bundled or differentiated pricing. It is sufficient to show that $x_j^\star$ is a convex function of the prices ${\bf p}$ for each user $j$; since ${\bf R_i x^\star}$ is an affine combination of convex functions, the lemma follows.  We first differentiate and find that
\begin{equation*}
\frac{\partial x_j^\star}{\partial p_k} = \frac{\gamma^{\frac{1}{1 - \alpha_j - \gamma}}}{1 - \alpha_j - \gamma}\left(\sum_{i = 1}^m p_i R_{ij}\right)^{\frac{\alpha_j + \gamma}{1 - \alpha_j -\gamma}}R_{kj} < 0
\end{equation*}
since $\gamma > 1 - \alpha_j$ by assumption. Upon differentiating a second time, we find that
\begin{equation*}
\frac{\partial^2 x_j^\star}{\partial p_k\partial p_l} = \frac{\left(\alpha_j + \gamma\right)\gamma^{\frac{1}{1 - \alpha_j - \gamma}}}{\left(1 - \alpha_j - \gamma\right)^2}\left(\sum_{i = 1}^m p_i R_{ij}\right)^{\frac{2\alpha_j + 2\gamma - 1}{1 - \alpha_j -\gamma}} R_{kj}R_{lj}.
\end{equation*}
Thus, the second-derivative matrix of $x_j^\star$ equals ${\bf R}{\bf QR}^T$, where ${\bf Q}$ is a diagonal matrix with entries
\begin{equation*}
Q_{jj} = \gamma^{\frac{1}{1 - \alpha_j - \gamma}}\frac{\alpha_j + \gamma}{\left(1 - \alpha_j - \gamma\right)^2}\left(\sum_{i = 1}^m p_i R_{ij}\right)^{\frac{2\alpha_j + 2\gamma - 1}{1 - \alpha_j -\gamma}}.
\end{equation*}
A sufficient condition for the second-derivative matrix to be positive semi-definite is then that
\begin{equation*}
\gamma^{\frac{1}{1 - \alpha_j - \gamma}}\frac{\alpha_j + \gamma}{\left(1 - \alpha_j - \gamma\right)^2}\left(\sum_{i = 1}^m p_i R_{ij}\right)^{\frac{2\alpha_j + 2\gamma - 1}{1 - \alpha_j -\gamma}} > 0,
\end{equation*}
which holds for each user $j$.
\end{proof}
We have thus shown that in order to optimize revenue, the operator should offer the lowest prices possible. However, in doing so the operator could inadvertently favor one user over another, possibly introducing user dissatisfaction. For instance, suppose that one user $k$ requires a relatively large amount of one resource $l$ (e.g., this user runs extremely computationally-intensive jobs and requires a large number of CPU cores). Under resource pricing, the operator might then increase the price $p_l$ of this resource slightly, in order to offer much lower prices for the other resources while keeping the total demand for resource $l$, $\sum_{j = 1}^n R_{lj}x_j^\star\left(\sum_{i = 1}^m R_{ij}p_i\right)$, beneath the capacity $C_l$. But in this scenario, user $k$ would then have a higher per-job cost $\sum_{i = 1}^m R_{ik}p_i$ than other users, and thus process a lower number of jobs and receive lower utility $U_j(x_j^\star) - r_jx_j^\star$. In the next section, we explore ways to limit inequality among users by measuring the fairness resulting from a set of offered prices.

\subsection{Fairness}\label{sec:fairness}

When users submit jobs to the datacenter operator, they are allocated a given amount of each resource, which suffices for the job to be completed. Thus, the operator is performing a \emph{multi-resource allocation} among users, as considered in \cite{ghodsi2011dominant,joe2012multi}. In these works, the fairness of a resource allocation is measured in terms of the number of jobs processed for each user--users are assumed to receive utility not from the resources allocated, but rather from the jobs that those resources allowed them to complete. In our framework, since users use utility functions to measure their utility from processed jobs, we measure the fairness of the utility value received (i.e., $\overline{U}_j = U_j\left(x_j^\star\right) - r_j {x_j^\star}^\gamma$) by each user $j$, rather than simply measuring the number of jobs processed for different users.

When measuring the fairness of a distribution, we consider two factors: equitability and efficiency. Equitability refers to the idea that users should receive equitable utility values, e.g., instead of assigning prices so that one user processes a much larger number of jobs and receives a higher utility than another user. However, in some cases one can increase a user's utility level without significantly diminishing other users' utilities: a more unequal allocation can result in more \emph{efficiency}, or a larger amount of total utility across users. Following \cite{joe2012multi}'s exploration of this equitability-efficiency tradeoff for multi-resource allocations, we utilize the following fairness functions:
\begin{equation}
{\rm sgn}(1 - \beta)\left(\sum_{j = 1}^n \overline{U}_j^{1 - \beta}\right)^{\frac{1}{\beta}}\left(\sum_{j = 1}^n \overline{U}_j\right)^{\lambda + 1 - \frac{1}{\beta}},
\label{eq:betalambda}
\end{equation}
where $\beta\neq 1$ and $\lambda$ respectively parameterize the equitability and efficiency \cite{lan2010axiomatic}.\footnote{These functions are the unique fairness functions satisfying four desirable fairness axioms \cite{lan2010axiomatic}. In Appendix \ref{sec:properties}, we consider two additional properties (envy-freeness and Pareto-efficiency) that are not axiomatically satisfied.} The term $\left(\sum_{j = 1}^n \overline{U}_j\right)^{\lambda}$, a function simply of the total utility across users, may be regarded as the efficiency component of this fairness function; as $\left|\lambda\right|$ increases, the total utility is more heavily weighted in the function. The remainder of the function, parameterized by a factor $\beta$, measures the equitability of the utility distribution. This component is homogeneous in the total utility value, and thus measures only the relative values of the utilities received by each user. We note that if $\lambda = \beta^{-1} - 1$, then maximizing (\ref{eq:betalambda}) is equivalent to maximizing
\begin{equation}
F_\beta({\bf p}) = \frac{1}{1 - \beta}\sum_{j = 1}^n \overline{U}_j^{1 - \beta},
\label{eq:betafairness}
\end{equation}
where $\beta > 0$ parameterizes the type of fairness under consideration. This fairness function corresponds to the well-known $\alpha$-fairness function with $\alpha = \beta$; we thus observe that as $\beta$ grows, (\ref{eq:betafairness}) becomes ``more fair'' \cite{lan2010axiomatic}.% \footnote{The $\alpha$-fairness function used here is \emph{independent} of the $\alpha$-fairness functions for user utility used in (\ref{eq:alphafair}); the present function refers to a distribution across users, while (\ref{eq:alphafair}) relates only to individual users' actions to optimize their utility.}

We next show that when users have isoelastic utility functions, the fairness function (\ref{eq:betafairness}) becomes a weighted fairness function of the number of jobs processed for each user. To do so, we first note the following lemma:
\begin{lma}\label{lem:const}
Given a set of prices ${\bf p}$, then assuming users choose their demand so as to maximize their isoelastic utilities (\ref{eq:alphafair}),
\begin{align}
\overline{U}_j &= \left(\frac{\gamma}{c_j}\right)^{\frac{\gamma}{1 - \alpha_j - \gamma}}\left(\frac{\gamma}{1 - \alpha_j} - 1\right)r_j^{\frac{1 - \alpha_j}{1 - \alpha_j - \gamma}} \nonumber \\
&= \left(\frac{\gamma}{1 - \alpha_j} - 1\right)r_j{x_j^\star}^\gamma.
\label{eq:constfairrev}
\end{align}
Thus, the utility received from each user equals the amount paid $r_j{x_j^\star}^\gamma$, multiplied by a price-independent constant.
\end{lma}
By choosing the prices to offer, the operator thus determines both the number of jobs processed $x_j^\star$ as well as the coefficient $\left(\gamma/(1 - \alpha_j) - 1\right)r_j$ of ${x_j^\star}^\gamma$ in users' total utility. If $\gamma = 1$, this coefficient simply multiplies $x_j^\star$, or the number of jobs processed by user $j$.

The fairness functions (\ref{eq:betalambda}) and (\ref{eq:betafairness}) assume that all users are equal, in the sense that the utility received by each user is treated the same. However, in practice users are heterogeneous--they each submit a distinct type of jobs to the datacenter operator, with different resource requirements and different utility functions. Thus, it may not be equitable to treat all users equally.  In Lemma \ref{lem:const}, we account for users' differing utility functions with the constant $\gamma/(1 - \alpha_j) - 1$. This constant weight is decreasing in $\alpha_j$; thus, as a user becomes less price-sensitive ($\alpha_j$ decreases), the weight in the fairness function increases. More price-sensitive users tend to receive lower utility values for the same number of jobs due to lower values of $\alpha_j$ in their utility functions (\ref{eq:alphafair}), an effect that is offset with this weight.

By examining the value of the per-job cost $r_j$ for bundled, resource, and differentiated pricing, we can gain further insight into the interpretation of fairness for each of these pricing plans. Under bundled pricing, we find that $r_j = \mu_j^\gamma p$. Supposing that the ratio of resources in each bundle equals the ratio of resource capacities, we see from (\ref{eq:domshare}) that $\mu_j x_j^\star$ is simply user $j$'s dominant share; thus, if there is no volume discount ($\gamma = 1$), then the fairness weight $\left(\gamma/(1 - \alpha_j) - 1\right)r_j$ is simply the price $p$ multiplied by $1/(1 - \alpha_j) - 1$ and the user's dominant share. Thus, the fairness on users' utility values can be viewed as a (weighted) fairness on the amount of resources required by each user. Users whose jobs have larger resource requirements (larger $\mu_j$) may thus process fewer jobs--their jobs require more of the limited resources, thus detracting from other users' jobs. This \emph{fairness on dominant shares} metric has been previously proposed for general allocations without pricing \cite{ghodsi2011dominant,joe2012multi}.

If we instead consider differentiated pricing, we note that the weight $\left(\gamma/(1 - \alpha_j) - 1\right)r_j = \left(\gamma/(1 - \alpha_j) - 1\right)\overline{p}_j$ in (\ref{eq:constfairrev}) is completely determined by the operator; users' resource requirements only enter into the choice of $\overline{p}_j$ through the demand functions $x_j^\star$. This pricing scheme thus corresponds to a fairness function that does not as explicitly account for resource requirements.

Under resource pricing with no volume discount, we obtain a fairness metric that lies between the fairness with bundled and differentiated pricing. With resource pricing, we see from (\ref{eq:constfairrev}) that each $x_j^\star$'s weight in the fairness function (\ref{eq:betalambda}) becomes $\left(\gamma/(1 - \alpha_j) - 1\right)\sum_{i = 1}^m p_i R_{ij}$. Thus, the number of jobs is multiplied by a constant depending on $\alpha_j$, multiplied by a weighted sum of a user's resource requirements. Users with higher resource requirements will thus process fewer jobs. However, unlike bundled pricing's dominant shares, the exact weight for each resource is determined by  its price, which may be seen as a proxy for how ``valuable'' different resources are. More valuable resources will have higher prices and thus yield higher weights.

% As discussed briefly in Section \ref{sec:related}, the fairness of a particular distribution or allocation of resources is generally measured in abstract terms--a centralized controller is assumed to freely allocate a given resource to several users, and can then evaluate the resulting allocation using a particular fairness metric. In this section, we first introduce previously proposed fairness metrics for the allocation of multiple resources, in which the resources combine to produce jobs. We then formulate these fairness metrics in terms of the prices offered to users and show that the resulting fairness functions satisfy two desirable fairness properties: sharing incentive and Pareto-efficiency.

% The fairness of a particular allocation or distribution involves two components: equitability and efficiency. The first, equitability, refers to the idea that users should receive equitable amounts of the resource being distributed; for instance, it is more fair to split a resource equally between two users rather than giving everything to one user. However, in some cases a more unequal allocation can result in more \emph{efficiency}, i.e., a larger amount of resource being distributed. In a datacenter context, the operator might be able to give a large number of jobs to one user with relatively little resource requirements, while giving fewer jobs to other users. This tradeoff between 

\subsection{Fairness-Revenue Tradeoffs}\label{sec:tradeoff}

A datacenter operator wishes to optimize both revenue and fairness. In this section, we consider a weighted sum of the revenue (\ref{eq:rev_all}) introduced in Section \ref{sec:revenue} and the fairness (\ref{eq:betafairness}) introduced in Section \ref{sec:fairness}: 
\begin{equation}
\nu\rho({\bf p}) + F_\beta\left({\bf p}\right)
\label{eq:weighted}
\end{equation}
where $\nu > 0$ is the weight parameter determining the operator's relative emphasis on fairness or revenue. By optimizing a weighted sum of revenue and fairness, the operator can optimize both its revenue and the fairness of users' resulting utility distribution. We consider users with isoelastic utility functions (\ref{eq:alphafair}) and show that, when the revenue weight $\nu$ is sufficiently small, this optimization problem becomes convex. While this result indicates that an operator can more easily optimize fairness, possibly leading to a loss in revenue, we can adapt the proof of Lemma \ref{lem:const} to show that this tradeoff is limited.

We first derive conditions on the weight $\nu$ under which maximizing (\ref{eq:weighted}) subject to the resource constraints is a convex optimization problem:
\begin{thm}\label{prop:concave}
Suppose that $\beta > 1$ and that
\begin{align*}
0 \leq \nu \leq &\left(\frac{\gamma + \alpha_j - 1}{1 - \alpha_j}\right)^{1 - \beta}\gamma^{\frac{\beta\gamma}{\alpha_j + \gamma - 1} - 1}\big(\beta(1 - \alpha_j) - \gamma\big) \\
&\times \left(\max_i\left(\frac{C_i}{R_{ij}}\right)^{\frac{\beta(\alpha_j - 1)}{\gamma}}\right)
\end{align*}
for all users $j$. Then (\ref{eq:weighted}) is a concave function of the prices and, by Lemma \ref{lem:resourceconvex}, maximizing (\ref{eq:weighted}) subject to the resource constraints is a convex optimization problem.
\end{thm}
\begin{proof}
See Appendix \ref{sec:proofs}.
\end{proof}

If $\nu$ satisfies Prop. \ref{prop:concave}'s conditions, standard convex optimization solvers may be employed to solve for the optimal prices. Even when $\nu$ is larger than the bound given, however, we can adapt these algorithms to find the optimal prices, as shown in Algorithm \ref{alg:intpoint} for the interior-point algorithm. Intuitively, the interior point algorithm starts from a set of prices ${\bf p}$ and then uses Lagrange multipliers on the problem constraints to iterate the values of ${\bf p}$ so as to increase the objective while remaining within the problem constraints. In our case, differentiating $F_\beta$ for $\beta > 1$ shows that the objective (\ref{eq:weighted}) is decreasing in each price variable--thus, intuitively, the operator searches for the lowest feasible prices. Using this intuition, one can adapt standard proofs of convergence for the interior-point algorithm to show the following: 
\begin{thm}\label{prop:revenuemax}
Algorithm \ref{alg:intpoint} converges to a price vector ${\bf p}$ maximizing the objective (\ref{eq:weighted}).
\end{thm}
In the case of only one resource, the lowest feasible price for this resource is the price for which the resource constraint is tight; by Lemma \ref{lem:resourceconvex}, since the demands $x_j^\star$ are decreasing in the price variables, a unique price exists satisfying this property. Thus, we obtain the following corollary:
\begin{cor}\label{cor:bundle}
Suppose that the operator offers a bundled pricing plan. Then its objective (\ref{eq:weighted}) for any revenue weight $\nu$ is maximized at the lowest feasible price, which is the zero of the bundled resource constraint $\sum_{j = 1}^n p\mu_j^\gamma{x_j^\star}(p)^\gamma = 1$.
\end{cor}
We finally note that Algorithm \ref{alg:intpoint} holds only for fixed volume discounts $\gamma$. However, if the operator's optimization problem is convex, as in Prop. \ref{prop:concave}, then the optimal prices for each $\gamma$ can be determined efficiently. The operator can then perform a line search to determine the optimal volume discount. Section \ref{sec:discounts} numerically characterizes the achieved fairness and revenue for different volume discounts.

\begin{algorithm}[t]
\scriptsize
\KwData{Parameters $\alpha_j$, $\nu$, and $\gamma$, tolerance $\epsilon > 0$, update parameter $\mu > 1$.}
\KwResult{Optimized resource prices.}
Initialize $t \leftarrow 1$\;
\While{$t < m/{tol}$}{
Find ${\bf p_k^\star}(t)$ satisfying $f_t({\bf p}) = -t\nu\nabla\rho({\bf p}) -t\nabla F_\beta({\bf p})+ \sum_{i = 1}^m \frac{-{\bf \nabla x^\star}({\bf p}) {\bf R_i}^T}{{\bf R_i x^\star}\left(\sum_{i = 1}^m p_i R_{ij}\right) - C_i} = 0$ where $k$ indexes the number of solutions\;
Choose $k$ such that ${\bf p_k^\star}(t)$ minimizes $f_t(p) = -t\nu\rho(p) -t\nabla F_\beta({\bf p}) - \sum_{i = 1}^m \log\left(-{\bf R_ix^\star\left(\sum_{i = 1}^m p_i R_{ij}\right)} + C_i\right)$\;
Update ${\bf \overline{p}} \leftarrow {\bf p^\star_k}(t)$\;
Update $t\leftarrow \mu t$\;}
\caption{Interior-point optimization of (\ref{eq:weighted}).}
\label{alg:intpoint}
\end{algorithm}

The upper bound on the revenue weight $\nu$ in Prop. \ref{prop:concave} raises some concern that the operator may overly weight fairness and consequently lose revenue. However, when users have isoelastic utility functions, we can in fact bound the achieved revenue in terms of the achieved fairness when the fairness parameter $\beta > 1$; thus, even if the fairness term in (\ref{eq:weighted}) dominates the revenue term, the operator's revenue is still lower-bounded away from zero. Moreover, if $\beta < 1$, a class of fairness measures not included in Prop. \ref{prop:concave}, then the achieved fairness can be lower-bounded in terms of the achieved revenue. Thus, emphasizing fairness in the operator's objective need not lead to a significant loss of either revenue or of the fairness measured by metrics besides those in Prop. \ref{prop:concave}.
\begin{thm}\label{prop:bound}
Suppose that $\beta > 1$. The revenue $\rho\left({\bf p}\right)$ can then be lower-bounded in terms of the achieved fairness (\ref{eq:betafairness}):
\begin{equation*}
\rho({\bf p}) \geq \big(F_\beta({\bf p})(1 - \beta)\big)^{\frac{1}{1 - \beta}}\sum_{j = 1}^n \frac{1 - \alpha_j}{\gamma + \alpha_j - 1}.
\end{equation*}
On the other hand, if $\beta < 1$, then the fairness value $F_\beta({\bf p})$ may be lower-bounded in terms of the achieved revenue:
\begin{equation*}
F_\beta({\bf p}) \geq \frac{\rho\left({\bf p}\right)^{1 - \beta}}{1 - \beta}\left(\frac{\gamma}{1 - \alpha_k} - 1\right)^{1 - \beta},
\end{equation*}
$\alpha_k = \max_j \alpha_j$. Intuitively, a lower value of $\beta$ more emphasizes efficiency: $\lambda = 1/\beta - 1$ in (\ref{eq:betalambda}) is increasing in $\beta$. By Lemma \ref{lem:const}, the efficiency, or total utility received, is simply a scalar multiple of the revenue. A higher $\beta$, on the other hand, gives more emphasis on fairness, and thus cannot be lower-bounded by the revenue.
\end{thm}
\begin{proof}
See Appendix \ref{sec:proofs}.
\end{proof}

\section{Numerical Illustrations}\label{sec:numerical}

In this section, we use a six-hour workload trace from a Google cluster to illustrate the effects of resource capacity and volume discounts on a datacenter operator's achieved revenue and fairness for different pricing plans \cite{hellerstein2010cluster}. In particular, we find that increasing resource capacity allows improvements in both fairness and revenue, but that these results are highly dependent on the heterogeneity of the users' resource requirements. Varying the volume discount offered allows an increase in fairness, but results in decreased revenue and more leftover resources under resource and differentiated pricing. We find that for a range of resource capacities and volume discounts, all three pricing plans yield comparable revenues, but differentiated pricing permits a significantly higher fairness.

Our workload trace includes 9218 jobs divided into 176580 tasks; each task runs on a single machine within the cluster. The amount of memory (RAM) and fractional number of CPU cores taken up by each active task was recorded at five-minute intervals over 6 hours; these measurements were then passed through a linear transformation to ensure anonymity. For this reason, we omit units in our discussion here. We add each job's resource usage over all the time intervals recorded to find an average per-job CPU usage of 0.136 and memory usage of 0.182.

\begin{figure}[t]
\centering
\includegraphics[width = 0.45\textwidth]{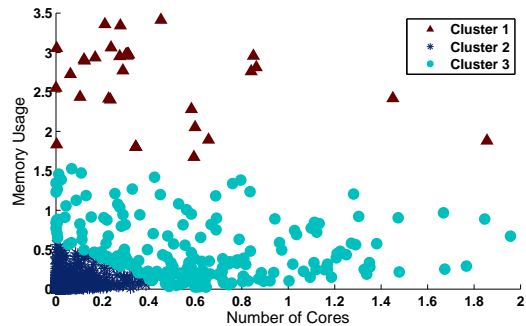}
\vspace{-0.05in}
\caption{Jobs in the Google trace, clustered based on the total CPU cores and memory used in a 6-hour interval.}
\label{fig:cluster}
\vspace{-0.15in}
\end{figure}

The jobs in the Google trace show a large variation in resource usage; the CPU and memory usage distributions have respective standard deviations of 13.4 and 18 times their means. To simplify our simulations, we exclude jobs whose total usage of either CPU or memory lies more than one standard deviation away from the mean. We then use $k$-means clustering to group jobs into three different clusters or types by their CPU and memory usage. Figure \ref{fig:cluster} shows the resulting clusters of jobs; we ran 30 instances of the $k$-means algorithm and chose the result minimizing the intra-cluster distance. We take the $k$-means centroid points to be the resource requirements of each job type: type 1 jobs require 0.4 CPUs and 2.7 units of memory, while type 2 jobs require only 0.01 units of CPU and 0.02 units of memory and type 3 jobs require 0.6 units of CPU and 0.5 units of memory. We associate a user type with each type of job and assume that these users have isoelastic utility functions (\ref{eq:alphafair}), with parameters $c_j = 1$ and $\alpha_j = 0.4, 0.7, 0.5$ for $j = 1,2,3$ respectively. Taking the capacity of each resource to be 6 units, the operator's objective (\ref{eq:weighted}) then becomes:
\begin{align}
\max_{\bf p}\;&\nu\rho({\bf p}) + F_\beta({\bf p}) \label{eq:opt}\\
{\rm s.t.}\;&\left\{0.4x_1^\star + 0.01 x_2^\star + 0.6 x_3^\star,\;2.7x_1^\star + 0.02x_2^\star + 0.5 x_3^\star\right\} \leq 6, \nonumber
\end{align}
where ${\bf p}$ is the vector of prices offered for either bundled, resource, or differentiated pricing. For scalability, we suppose that there are 10 users, with 1 user each of types 1 and 3 and 8 users of type 2.

\begin{figure*}
\centering
	\begin{subfigure}{0.34\textwidth}
		\centering
		\includegraphics[width = \textwidth]{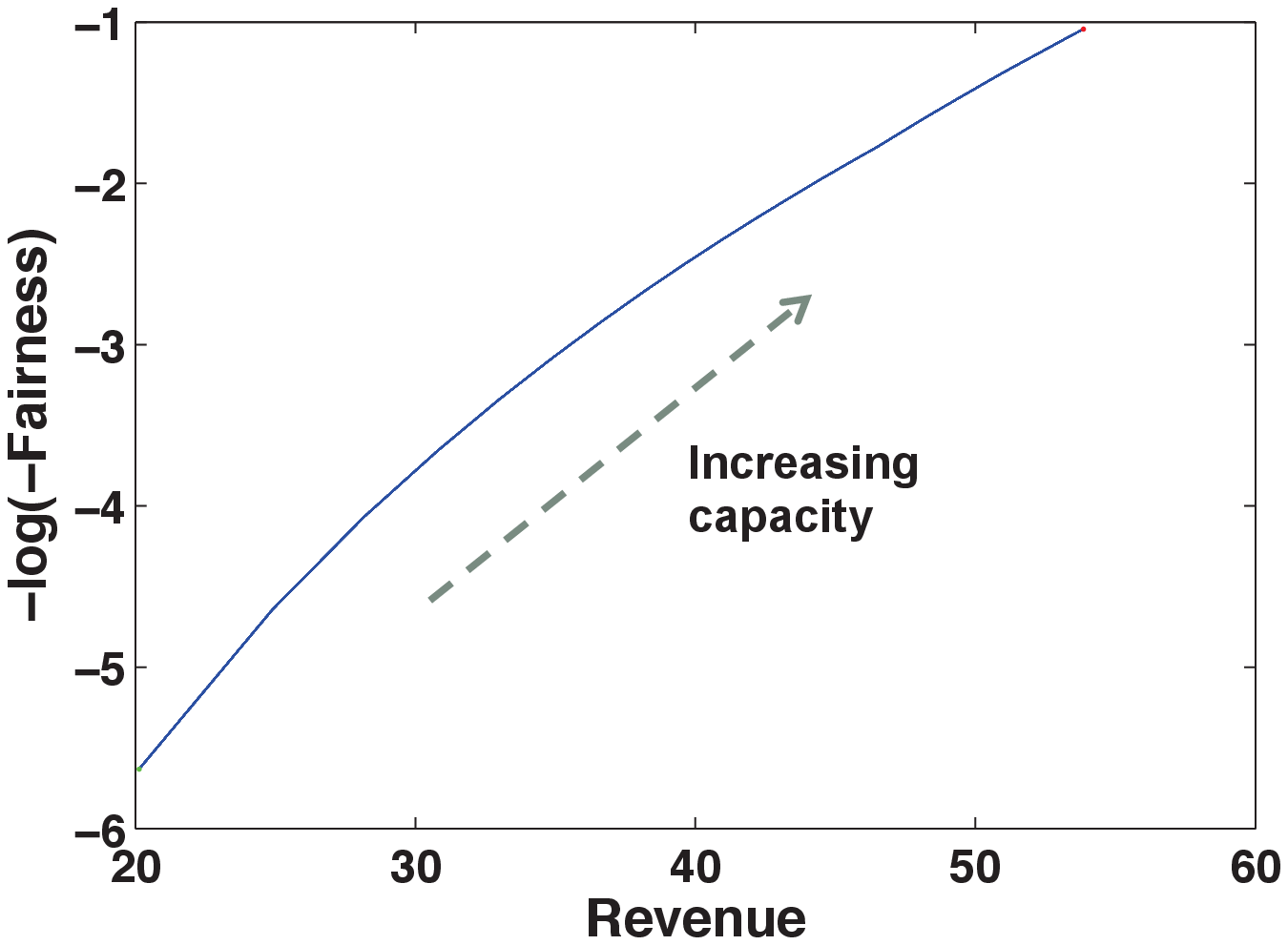}
		\caption{Bundled pricing.}
		\label{fig:capacity_bund_fair_rev}
	\end{subfigure}
	\hspace{-0.02\textwidth}
	\begin{subfigure}{0.34\textwidth}
		\centering
		\includegraphics[width = \textwidth]{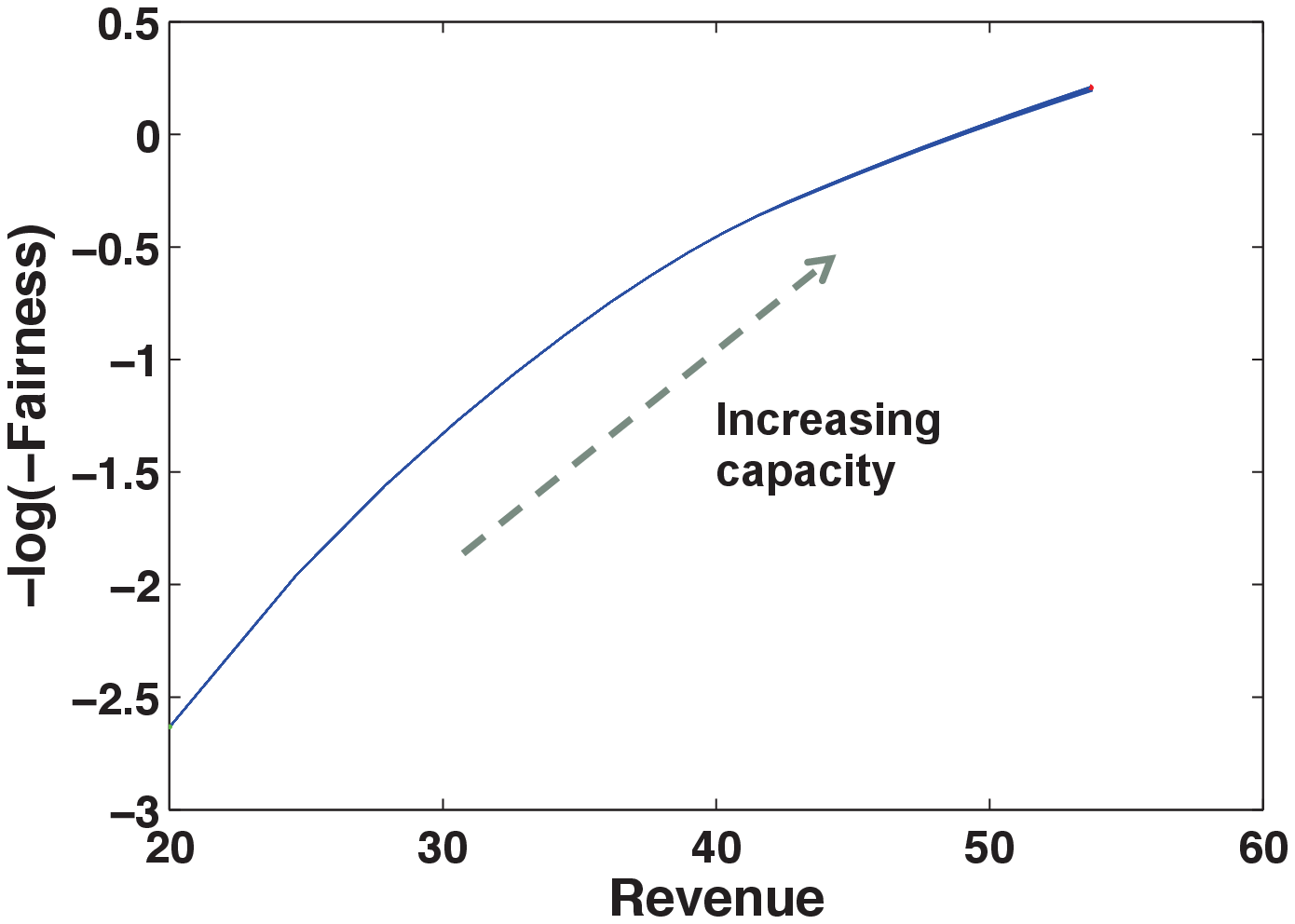}
		\caption{Resource pricing.}
		\label{fig:capacity_res_fair_rev}
	\end{subfigure}
	\hspace{-0.02\textwidth}
	\begin{subfigure}{0.34\textwidth}
		\centering
		\includegraphics[width = \textwidth]{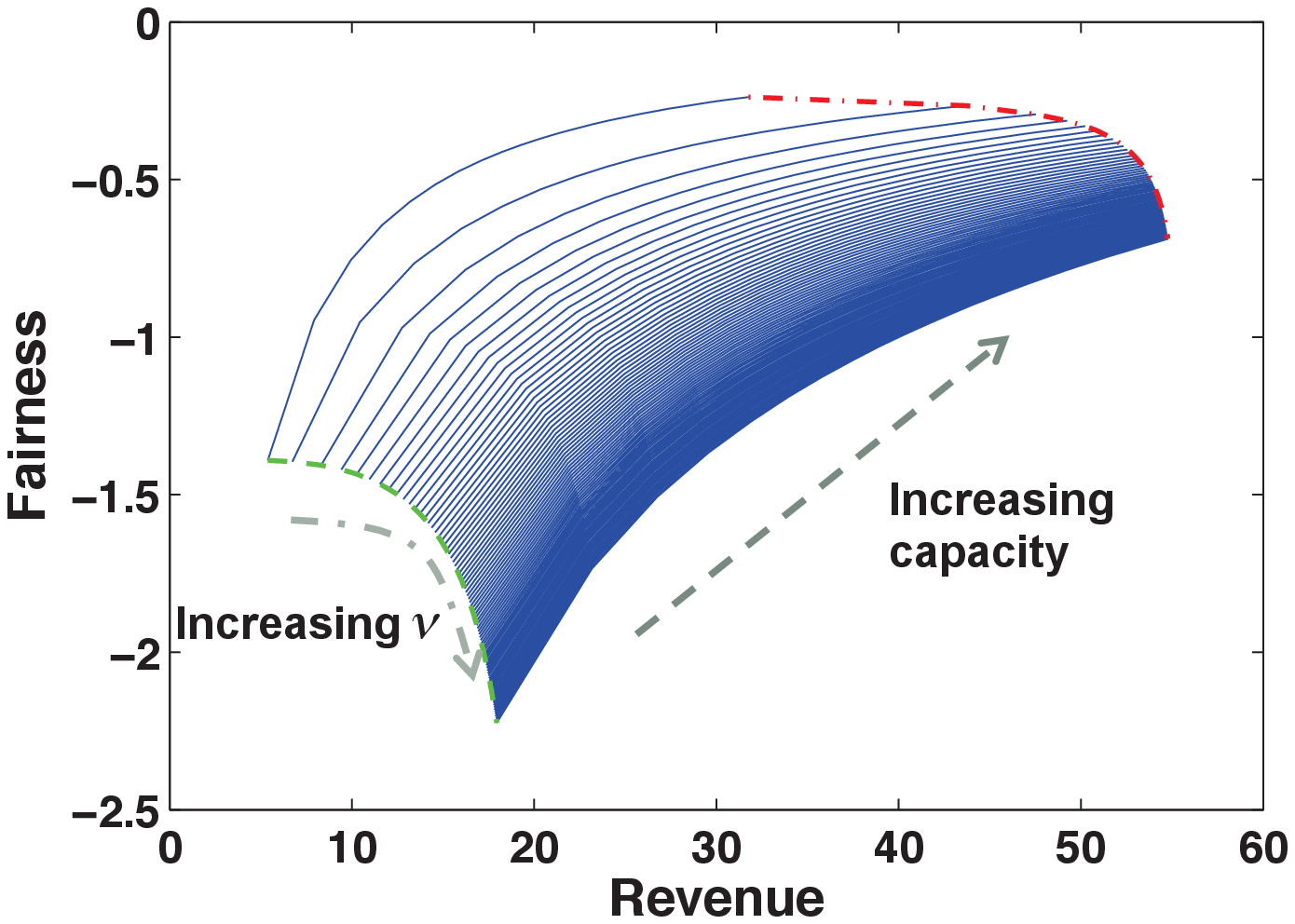}
		\caption{Differentiated pricing.}
		\label{fig:capacity_diffnu_fair_rev}
	\end{subfigure}
	\hspace{-0.02\textwidth}
	\vspace{-0.05in}
	\caption{Achieved fairness and revenue for a range of revenue weights $\nu$ in (\ref{eq:opt}) and memory capacities. On each contour (individual contours only visible for differentiated pricing), $\nu$ is fixed and the capacities are varied. For differentiated pricing, the contours reveal a range of tradeoffs for each fixed capacity and varied weight $\nu$.}
	\label{fig:capacity_fair_rev}
	\vspace{-0.05in}
\end{figure*}
\begin{figure*}[t]
\centering
	\begin{subfigure}{0.34\textwidth}
		\centering
		\includegraphics[width = \textwidth]{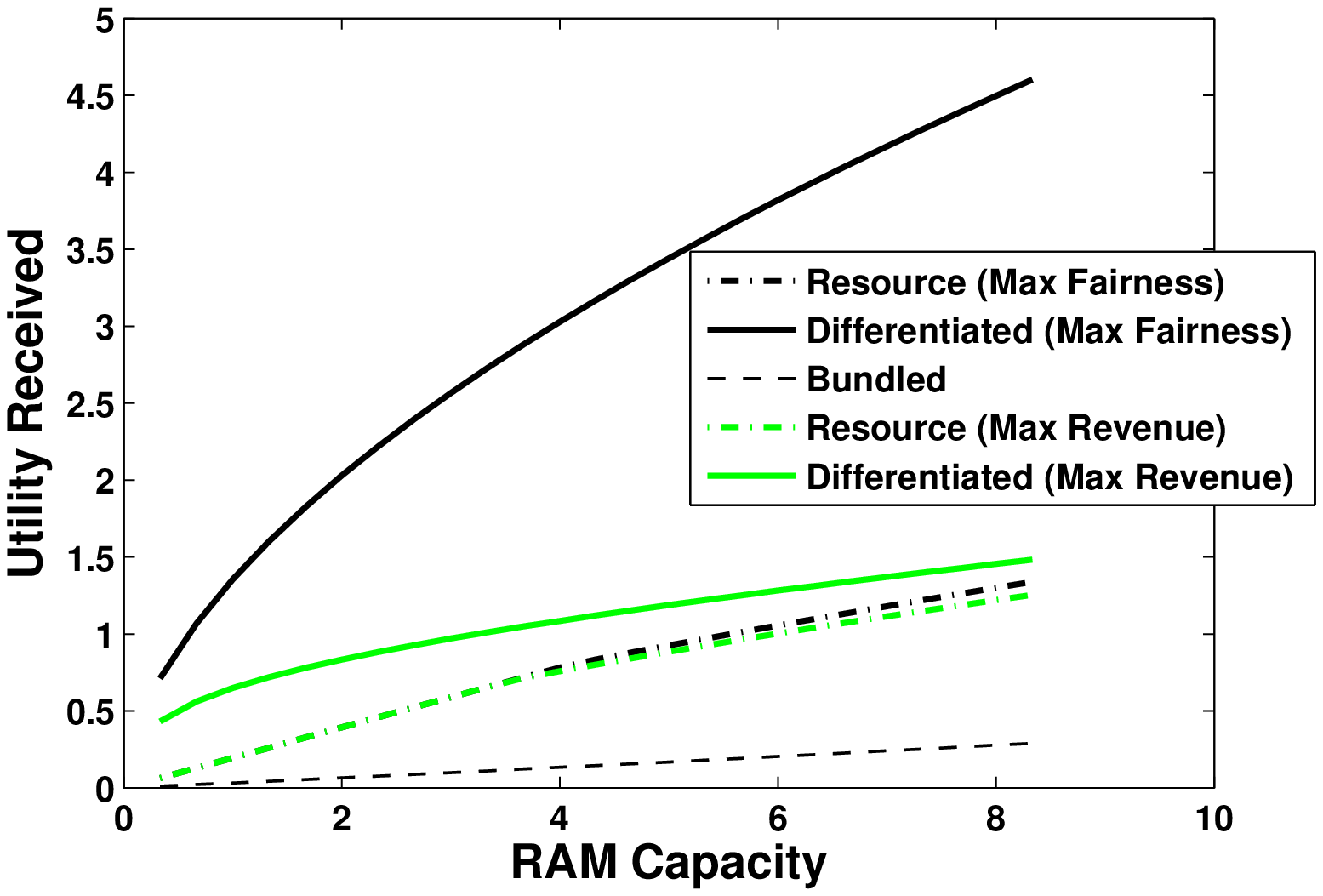}
%		\vspace{-0.05in}
		\caption{Type 1.}
		\label{fig:capacity_1_utility}
	\end{subfigure}
	\hspace{-0.02\textwidth}
	\begin{subfigure}{0.34\textwidth}
		\centering
		\includegraphics[width = \textwidth]{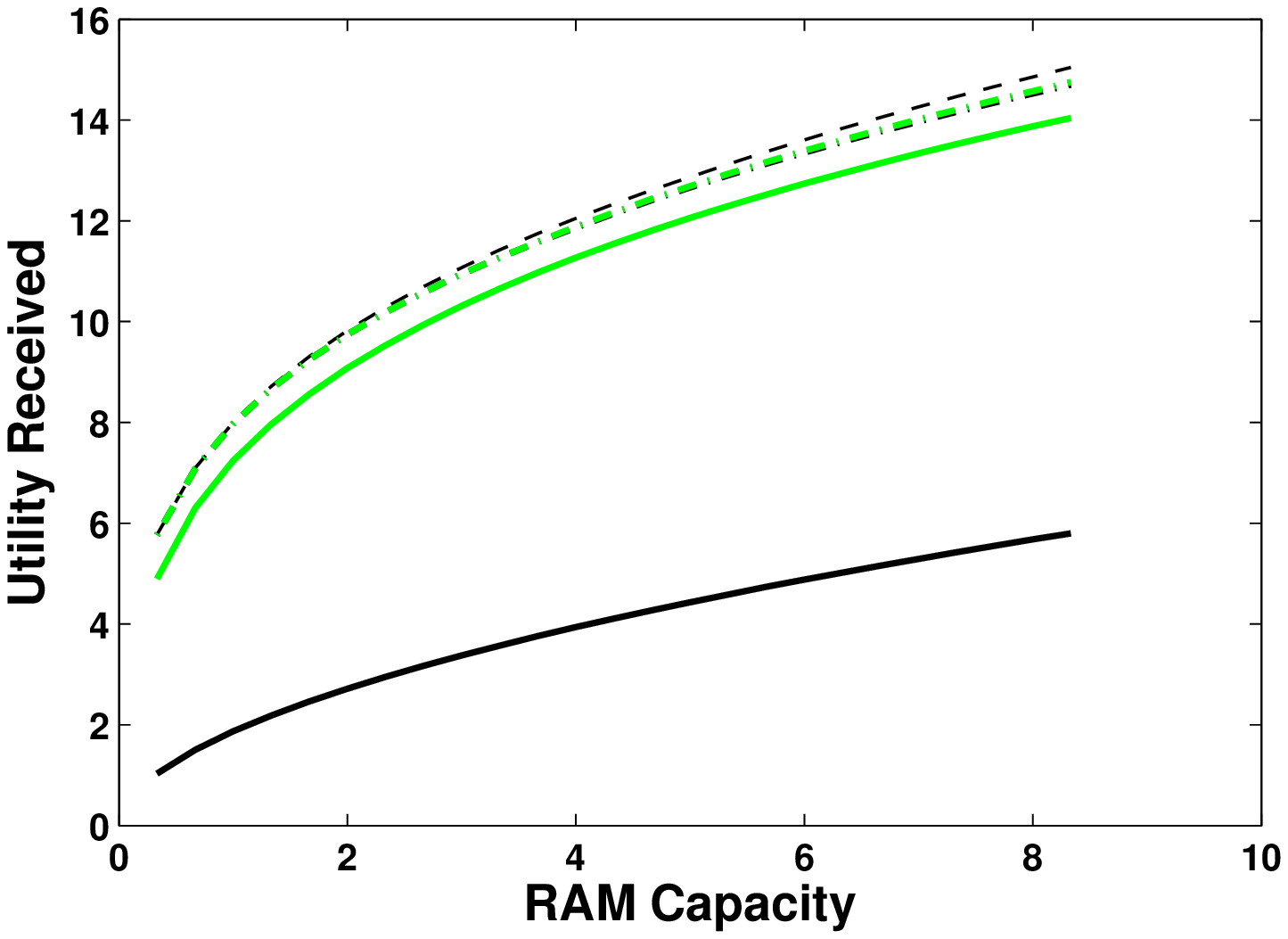}
%		\vspace{-0.05in}
		\caption{Type 2.}
		\label{fig:capacity_2_utility}
	\end{subfigure}
	\hspace{-0.02\textwidth}
	\begin{subfigure}{0.34\textwidth}
		\centering
		\includegraphics[width = \textwidth]{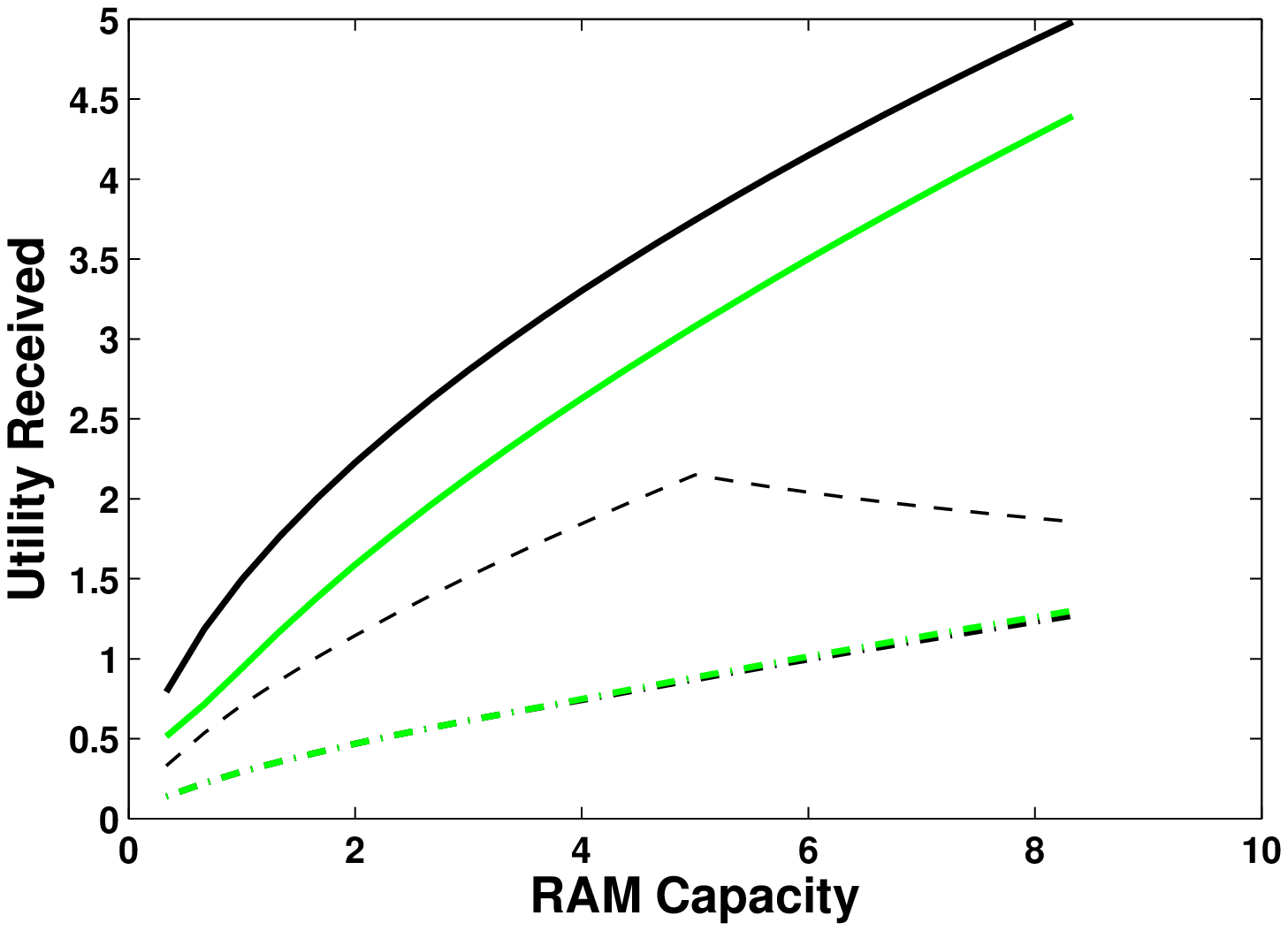}
%		\vspace{-0.05in}
		\caption{Type 3.}
		\label{fig:capacity_3_utility}
	\end{subfigure}
	\hspace{-0.02\textwidth}
	\vspace{-0.05in}
	\caption{Utility received by the three user types as the RAM (memory) capacity varies for maximum fairness or revenue.}
	\label{fig:capacity_utility}
	\vspace{-0.15in}
\end{figure*}

We first consider the effects of resource capacity on the operator's revenue and fairness in Section \ref{sec:capacity}, and then show that these tradeoffs can change dramatically with the distribution of users (Section \ref{sec:dist}). Section \ref{sec:discounts} then considers the effects of volume discounts. Throughout our discussion, we take the fairness parameter $\beta$ in (\ref{eq:opt}) to be relatively large, at $\beta = 20$; this choice of $\beta$ allows us to approximate max-min fairness, i.e., maximizing the minimum utility value.  Larger values of $\beta$ thus impose a stricter fairness requirement than lower ones, and have been used previously as fairness benchmarks in multi-resource allocation \cite{ghodsi2011dominant,joe2012multi}.

\subsection{Resource Capacity}\label{sec:capacity}

We first examine the effect of the available capacity on the operator's achieved revenue and fairness. We vary the memory capacity from 1/3 to 8 and use the interior-point algorithm ({\it cf}. Algorithm \ref{alg:intpoint}) to solve for the optimal values of the bundled, resource, and differentiated prices for different revenue weights $\nu$ in (\ref{eq:opt}). Figure \ref{fig:capacity_fair_rev} shows the achieved revenue and fairness for each pricing plan. As the memory capacity increases in Fig. \ref{fig:capacity_fair_rev}, the operator's fairness and efficiency under each pricing plan both improve: the operator has more memory to allocate, and thus more flexibility in setting the prices. Mathematically, we see that one constraint in the optimization problem (\ref{eq:opt}) has been relaxed, allowing the operator to increase the value of its objective function. We note that, as expected from Corollary \ref{cor:bundle}, the achieved revenue and fairness do not change with $\nu$ for bundled pricing. The optimal bundled price, which maximize both revenue and fairness, is simply the lowest price consistent with the operator's resource constraint.

%\begin{figure}
%	\centering
%	\includegraphics[width = 0.35\textwidth]{Matlab/Figures/Capacity_Diff_FixedParam_Fairness_Revenue.eps}
%	\caption{Achieved fairness and revenue under differentiated pricing. Capacity is fixed on each contour, while revenue weights $\nu$ are varied.}
%	\label{fig:capacity_diffparam_fair_rev}
%\end{figure}

We next compare resource (Fig. \ref{fig:capacity_res_fair_rev}) and differentiated (Fig. \ref{fig:capacity_diffnu_fair_rev}) pricing. Though differentiated pricing adds only a single additional price variable, it results in a significantly larger fairness for the operator, especially at low memory capacities. Moreover, varying the revenue weight $\nu$ results in a much richer set of tradeoffs between fairness and revenue for differentiated, rather than resource, pricing. This result may reflect the large heterogeneity in users' resource requirements: while both type 1 and type 2 users require a very large amount of memory relative to their CPU requirements, type 2 users require significantly less memory and fewer CPUs than type 1 users. If the operator offers resource pricing, type 1 users are then very affected by the memory price, and their utility will be adversely impacted by an increase in this price. In fact, type 1 and type 3 users also require more CPUs per job than type 2 users, so their utility levels will also be adversely impacted by increases in the CPU price. At the same time, type 2 users will process many jobs due to their jobs' relatively low resource requirements cost, allowing the operator to extract more revenue from type 2 users. Under differentiated pricing, the operator can manipulate the prices of individual users to remove these constraints imposed by resource heterogeneity, allowing the operator to increase the fairness across users without a large decreases in revenue.
% Thus, in order for the operator to increase fairness, the operator may set a low memory price so that users of type 1 can process more jobs. Due to type 2 users' low memory requirements, however, this price cannot be too low, or type 2 users' memory requirements will result in demand for memory exceeding capacity. As memory capacity increases, the operator has more room to maneuver, resulting in the observed tradeoffs for differentiated pricing. Moreover, the CPU price cannot be used to fully offset the low price for memory; type 3 users require nearly as much memory as CPUs, so for low CPU prices, type 3 users will exceed the CPU demands.

\begin{figure*}[t]
\centering
	\begin{subfigure}{0.34\textwidth}
		\centering
		\includegraphics[width = \textwidth]{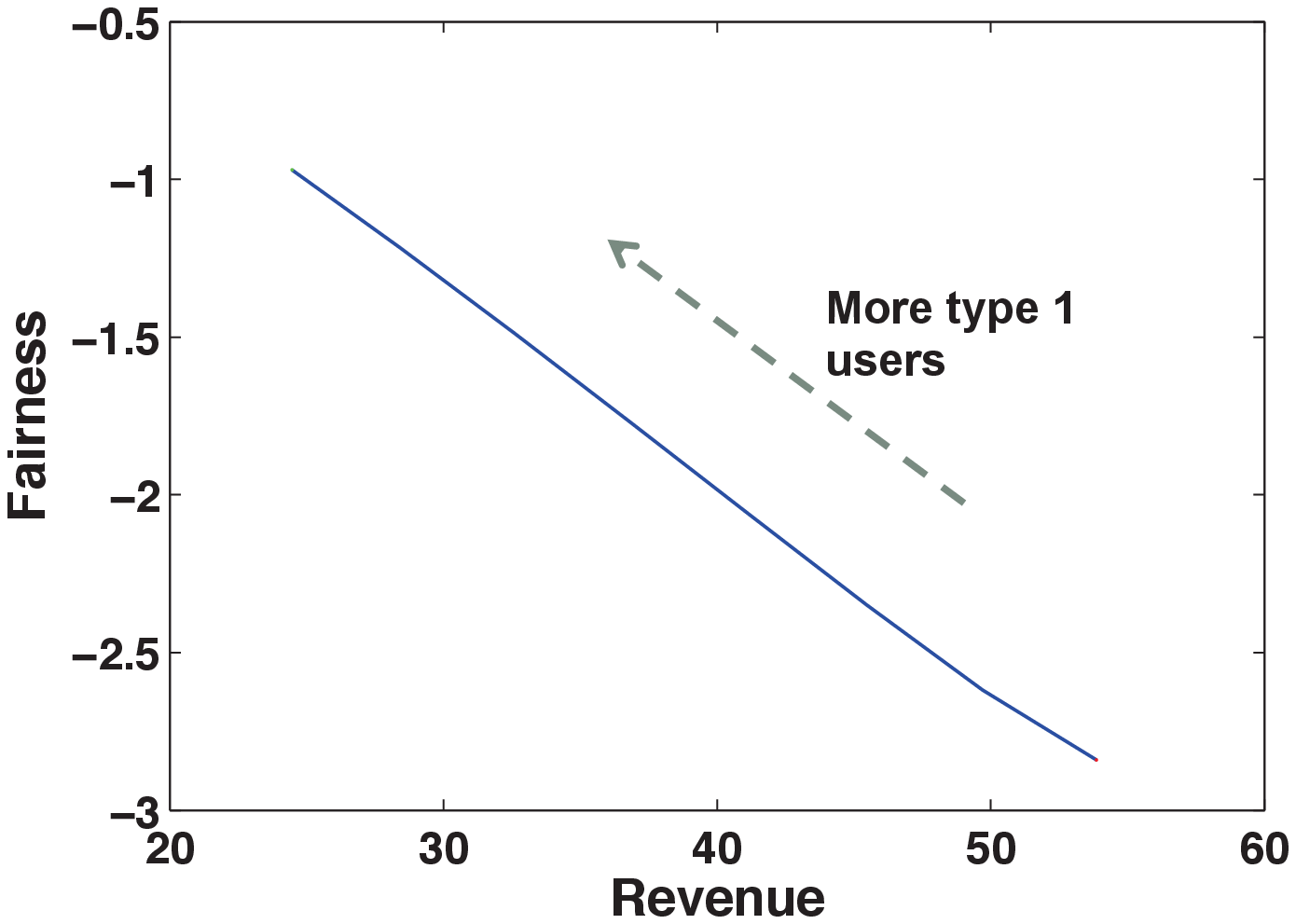}
%		\vspace{-0.05in}
		\caption{Bundled pricing.}
		\label{fig:type_bund_fair_rev}
	\end{subfigure}
	\hspace{-0.02\textwidth}
	\begin{subfigure}{0.34\textwidth}
		\centering
		\includegraphics[width = \textwidth]{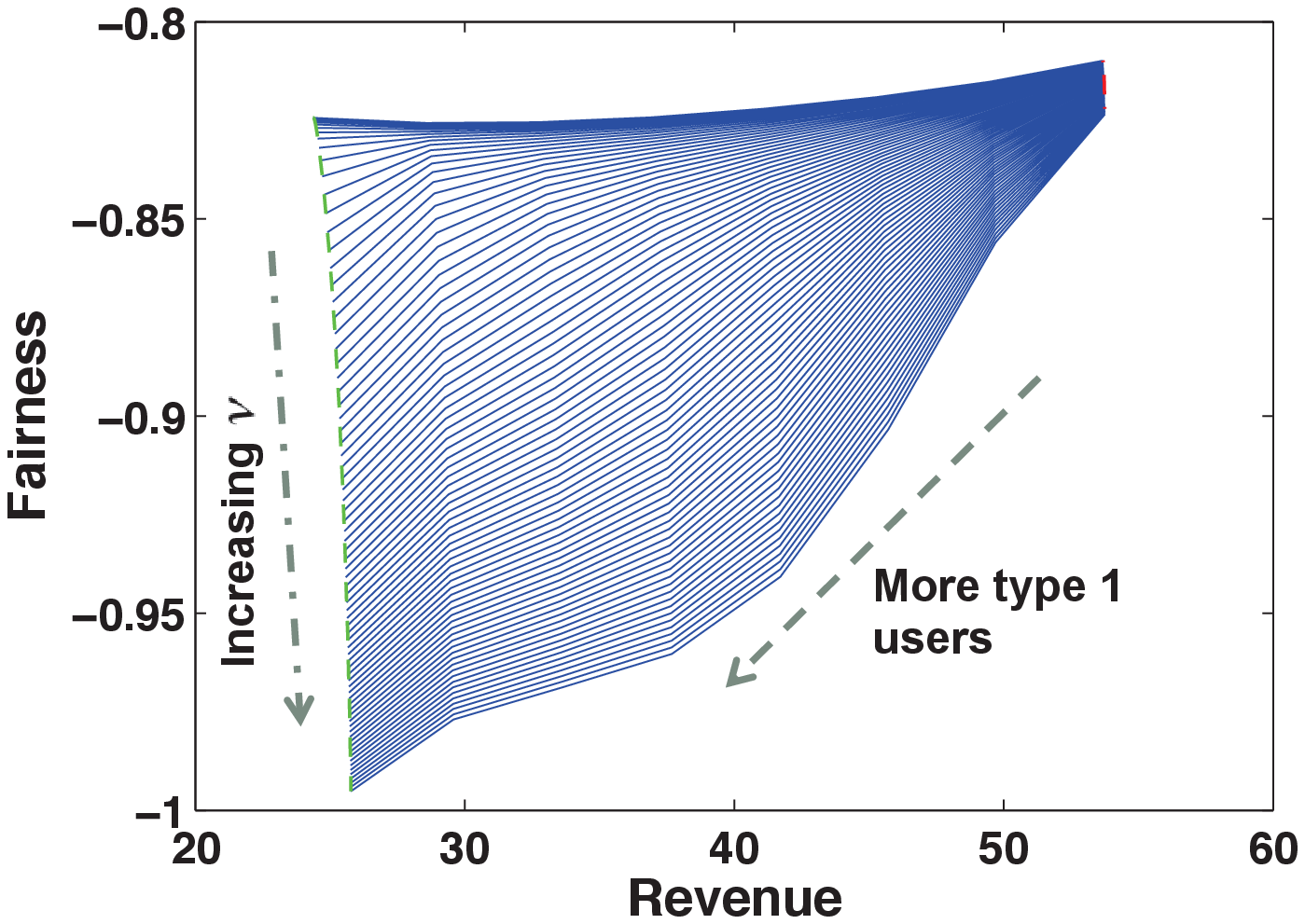}
%		\vspace{-0.05in}
		\caption{Resource pricing.}
		\label{fig:type_res_fair_rev}
	\end{subfigure}
	\hspace{-0.02\textwidth}
	\begin{subfigure}{0.34\textwidth}
		\centering
		\includegraphics[width = \textwidth]{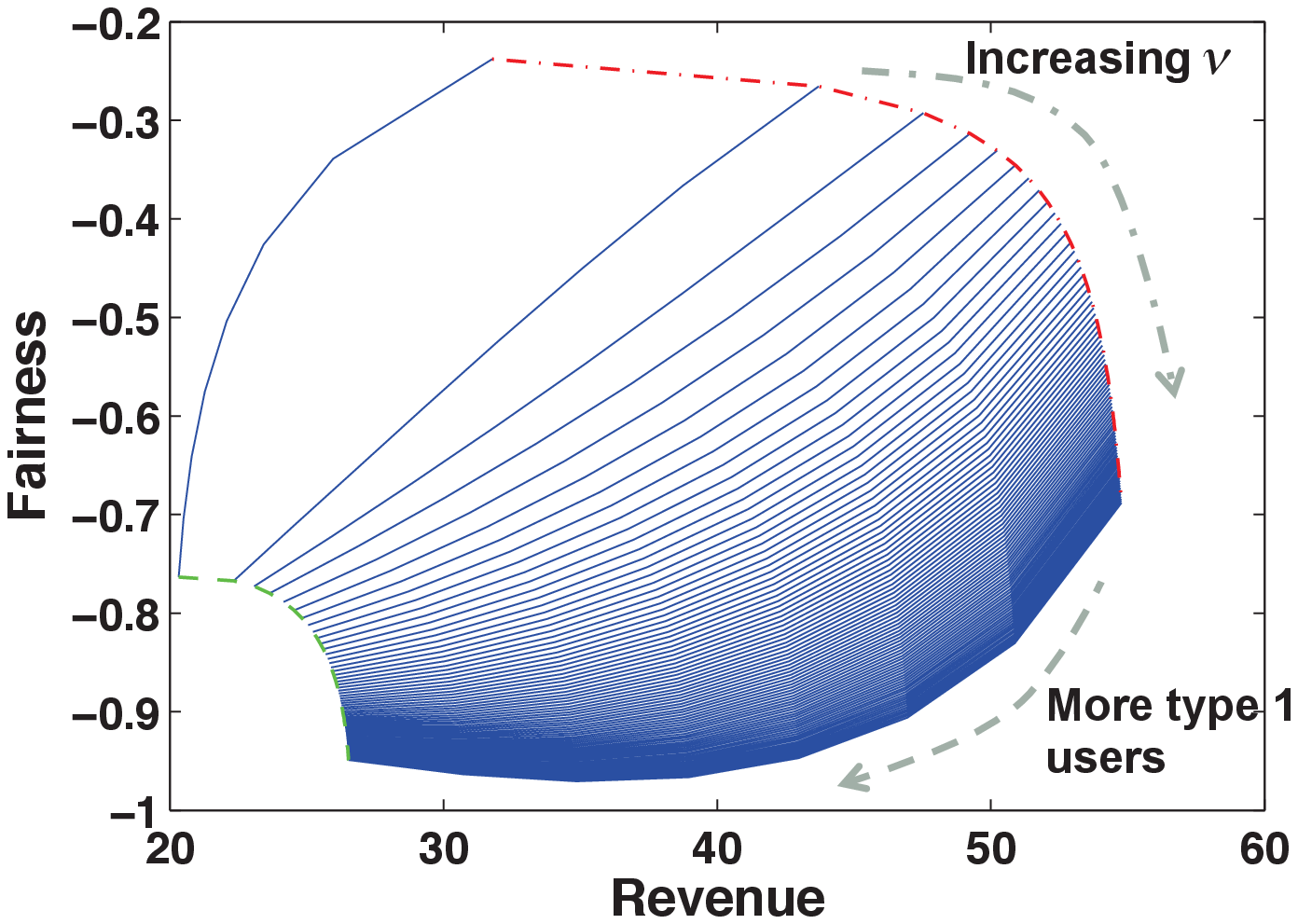}
%		\vspace{-0.05in}
		\caption{Differentiated pricing.}
		\label{fig:type_diff_fair_rev}
	\end{subfigure}
	\hspace{-0.02\textwidth}
	\vspace{-0.05in}
	\caption{Achieved fairness and revenue for a range of revenue weights $\nu$ and distribution of user types. On each contour (individual contours not visible in (a)), the weight $\nu$ is fixed and the fraction of type 1 users varies. For (b) resource and (c) differentiated pricing, the contours reveal a range of tradeoffs for each distribution of user types and varied revenue weights $\nu$.}
	\vspace{-0.05in}
	\label{fig:type_fair_rev}
\end{figure*}

We can validate the above explanation of the operator's low achieved fairness for bundled and resource pricing by calculating the utility received by each user under the three pricing plans. Figure \ref{fig:capacity_utility} shows the utility values for each user when only fairness is optimized (i.e., $\nu = 0$ in (\ref{eq:opt})), and when only revenue is optimized ($\nu\rightarrow\infty$). We see that type 2 users receive more utility than users of types 1 or 3; since type 2 users have far smaller resource requirements, this result is unsurprising. Users 1 and 3, however, receive almost the same utility under bundled and resource pricing, which is much less than that under differentiated pricing. As the capacity increases, however, their utility levels under resource pricing become visibly larger than those under bundled pricing, corroborating Fig. \ref{fig:capacity_fair_rev}'s higher fairness values under resource, rather than bundled, pricing.

\subsection{User Heterogeneity}\label{sec:dist}

We next examine whether the dramatic changes in fairness for different pricing plans is still observed when users are less heterogeneous. We vary the fraction of type 1 and type 2 users varies (we fix the fraction of type 3 users to be 10\%) and fix the memory capacity to be 6 units. Figure \ref{fig:type_fair_rev} shows the resulting fairness and revenue tradeoffs; we see that a visible range of tradeoffs exists for both resource and differentiated pricing, and that the fairness for different revenue weights $\nu$ is much less negative than in Fig. \ref{fig:capacity_fair_rev}. As the fraction of type 1 users increases, the revenue decreases under all pricing plans: type 1 users require more resources than type 2 users, so the operator receives more revenue from using its available resources to process type 2 users' jobs. Under bundled pricing, the fairness also increases with the fraction of type 1 users; all type 1 users receive the same (lower than type 2) utility values, increasing equitability. Under resource or differentiated pricing, however, increasing the fraction of type 1 users decreases the fairness. To achieve a higher revenue, the operator processes more jobs from type 2 users, even when more type 1 users are present.

\begin{figure*}[t]
\centering
	\begin{subfigure}{0.34\textwidth}
		\centering
		\includegraphics[width = \textwidth]{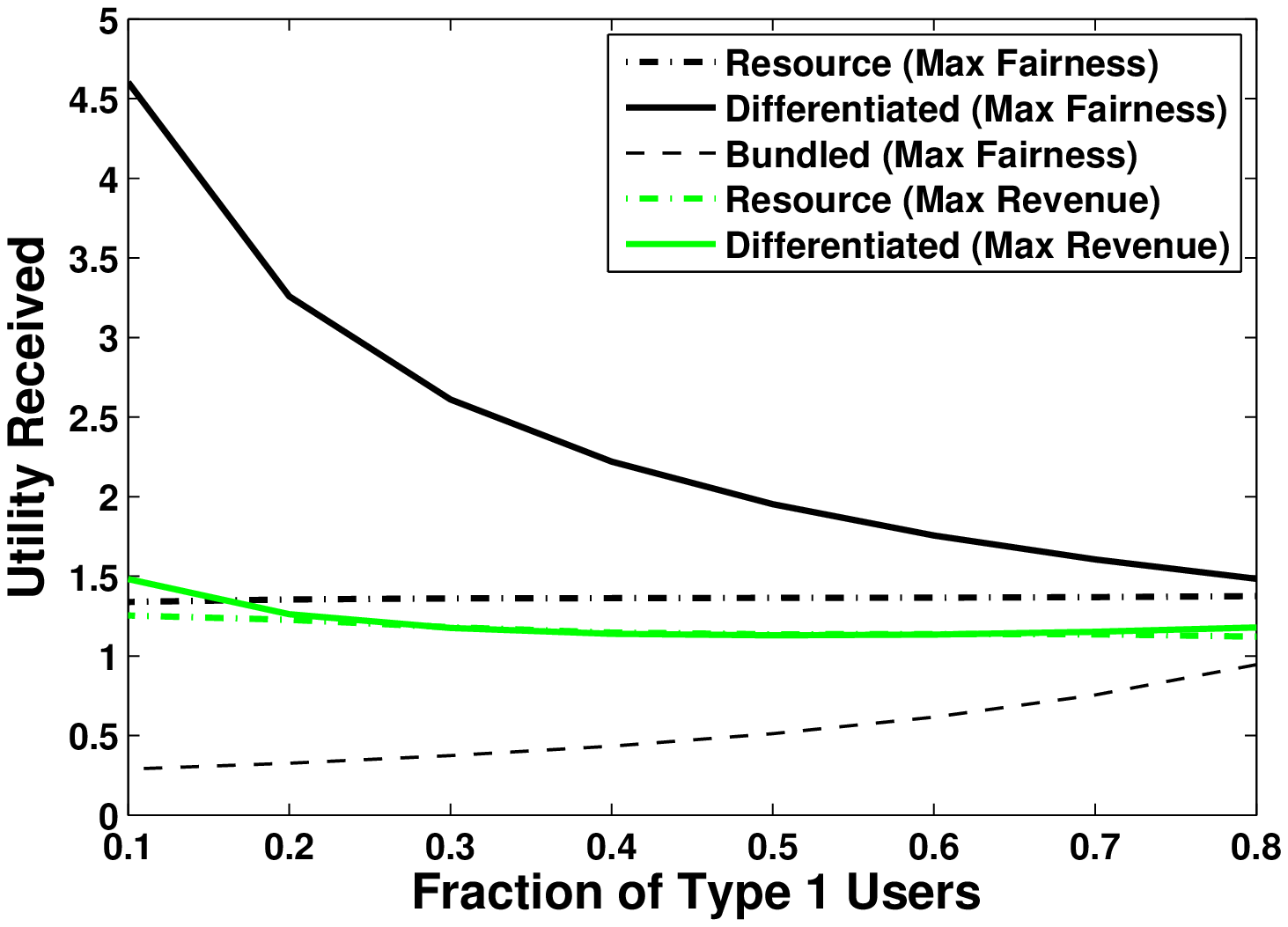}
		\caption{Type 1.}
		\label{fig:type_1_utility}
	\end{subfigure}
	\hspace{-0.02\textwidth}
	\begin{subfigure}{0.34\textwidth}
		\centering
		\includegraphics[width = \textwidth]{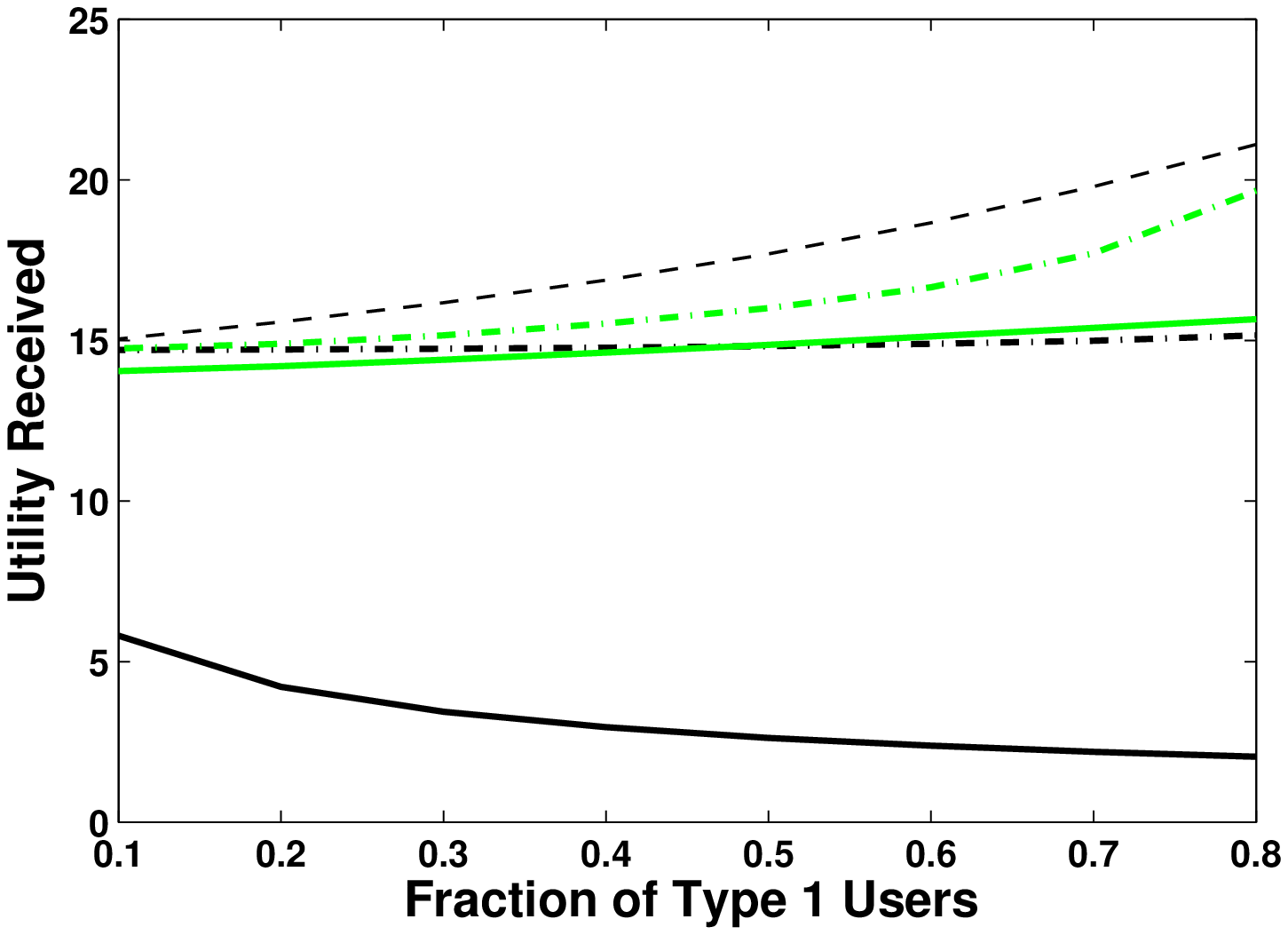}
		\caption{Type 2.}
		\label{fig:type_2_utility}
	\end{subfigure}
	\hspace{-0.02\textwidth}
	\begin{subfigure}{0.34\textwidth}
		\centering
		\includegraphics[width = \textwidth]{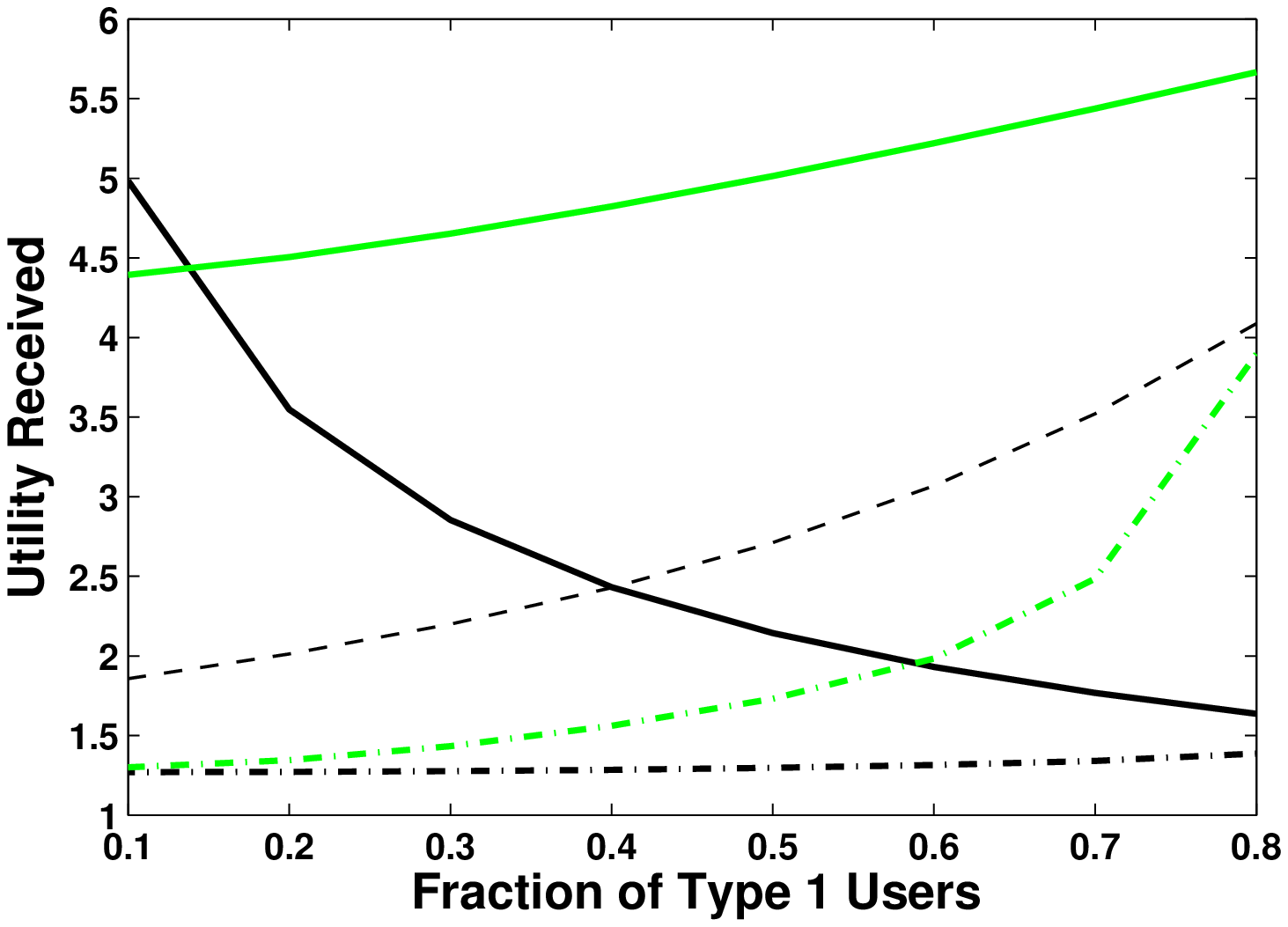}
		\caption{Type 3.}
		\label{fig:type_3_utility}
	\end{subfigure}
	\hspace{-0.02\textwidth}
	\vspace{-0.05in}
	\caption{Utility received by the three user types as the fraction of type 1 and 2 users varies, maximum fairness or revenue.}
	\label{fig:type_utility}
	\vspace{-0.1in}
\end{figure*}

As expected from Corollary \ref{cor:bundle}, varying the revenue weight $\nu$ has no effect on the utility or revenue from bundled pricing. Under resource or differentiated pricing, however, as the revenue weight $\nu$ increases in Fig. \ref{fig:type_fair_rev}, the operator's revenue increases for any given fraction of type 1 users. However, this increase in revenue results in a decrease in fairness: in particular, when a low fraction of users is present in resource pricing (e.g., the leftmost contour in Fig. \ref{fig:type_res_fair_rev}), the revenue barely increases, while the fairness drops steeply as $\nu$ increases. We can explain this result by examining the utility received by each user for different fractions of type 1 users, as shown in Fig. \ref{fig:type_utility}.

\begin{figure*}[t]
\centering
	\begin{subfigure}{0.34\textwidth}
		\centering
		\includegraphics[width = \textwidth]{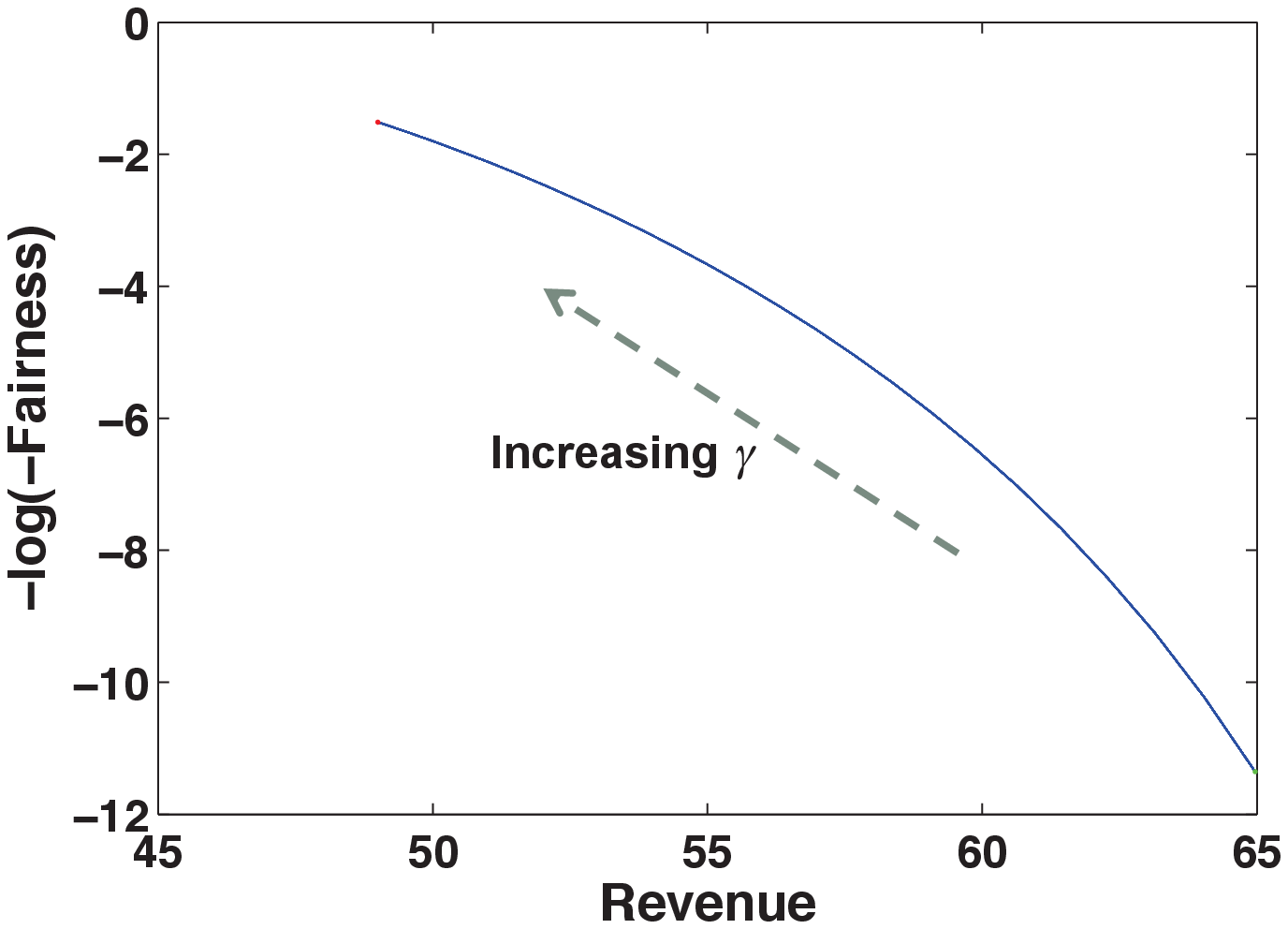}
%	\vspace{-0.05in}
		\caption{Bundled pricing.}
		\label{fig:gamma_bund_fair_rev}
	\end{subfigure}
	\hspace{-0.02\textwidth}
	\begin{subfigure}{0.34\textwidth}
		\centering
		\includegraphics[width = \textwidth]{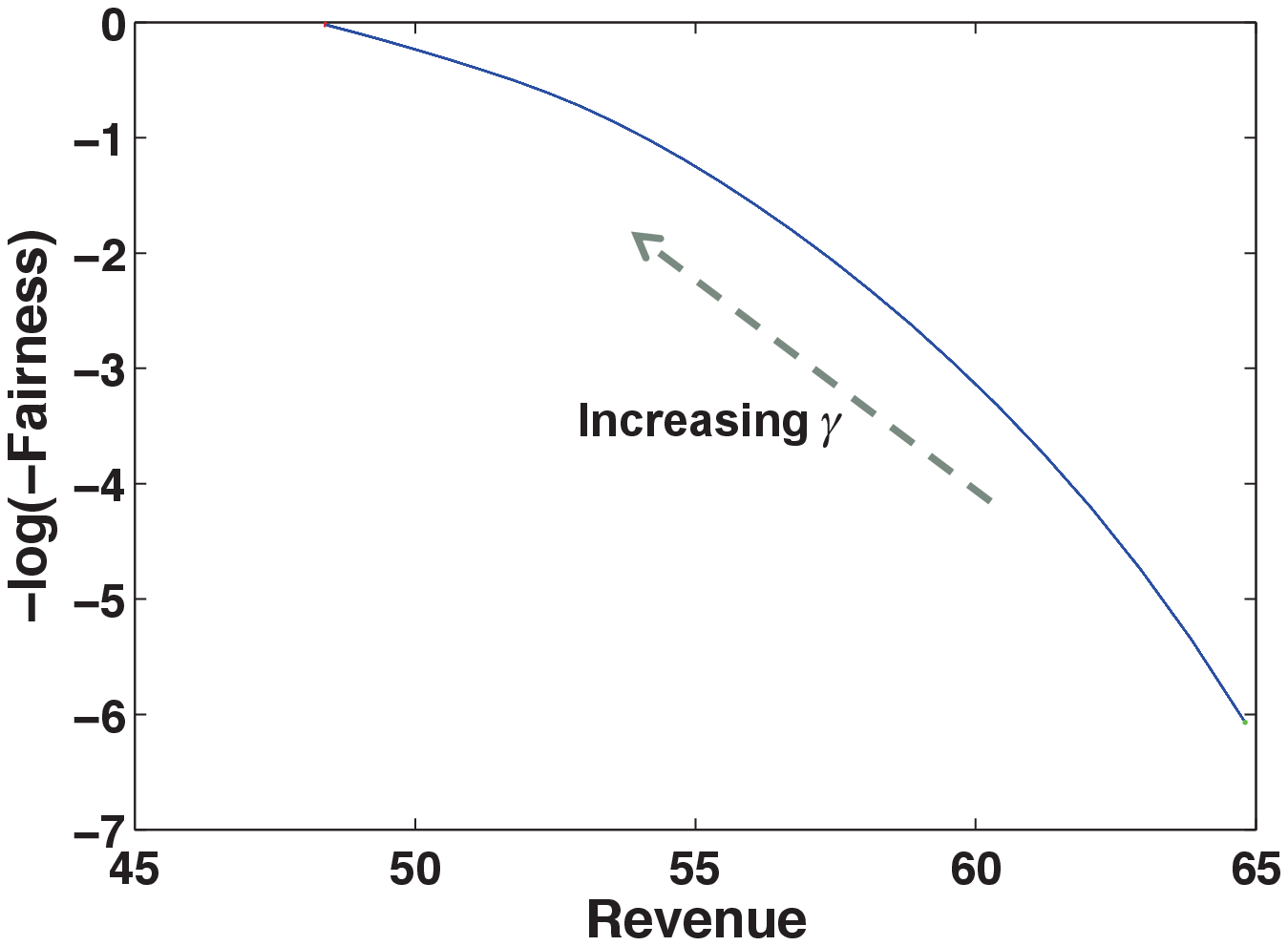}
%	\vspace{-0.05in}
		\caption{Resource pricing.}
		\label{fig:gamma_res_fair_rev}
	\end{subfigure}
	\hspace{-0.02\textwidth}
	\begin{subfigure}{0.34\textwidth}
		\centering
		\includegraphics[width = \textwidth]{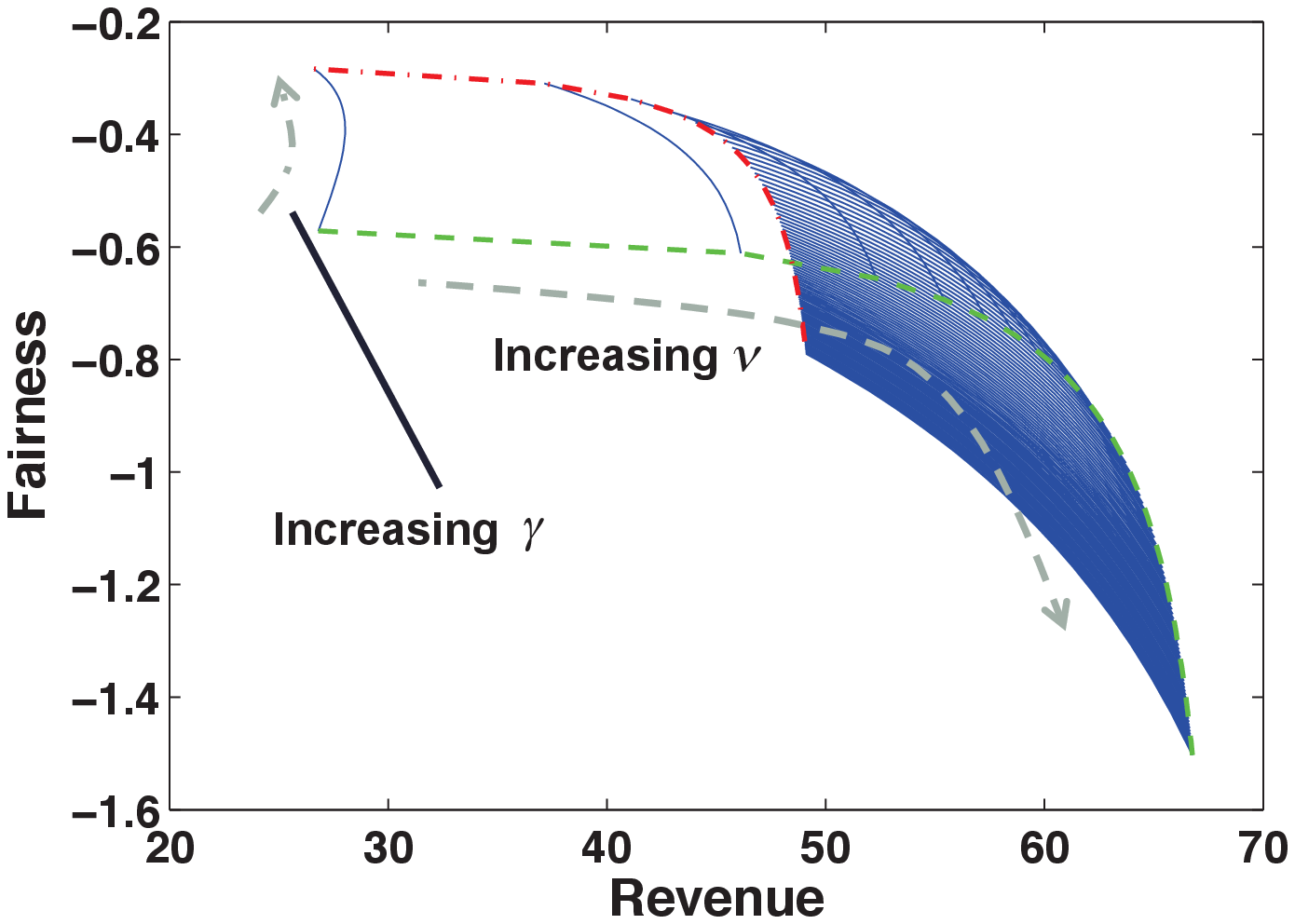}
%	\vspace{-0.05in}
		\caption{Differentiated pricing.}
		\label{fig:gamma_diff_fair_rev}
	\end{subfigure}
	\hspace{-0.02\textwidth}
	\vspace{-0.05in}
	\caption{Achieved fairness and revenue for a range of revenue weights $\nu$ and volume discounts $\gamma$. On each contour (individual contours only visible for differentiated pricing), $\nu$ is fixed and $\gamma$ varies. For (a) bundled and (b) resource pricing, the optimal fairness and revenue values do not change significantly with $\nu$. For (c) differentiated pricing, the contours reveal a range of tradeoffs for each fixed $\gamma$ and varied $\nu$.}
	\label{fig:gamma_fair_rev}
	\vspace{-0.15in}
\end{figure*}

We see from Fig. \ref{fig:type_utility} that as the fraction of type 1 users increases, the minimum utility over all three users also decreases, though type 2 and type 3 users actually receive more utility under bundled and resource pricing when the revenue is maximized. This is offset, however, by a decrease in utility for type 1 users, which results in a decrease in the minimum utility across users. This result likely comes from the fact that type 1 users require more resources to process their increased number of jobs, as the fraction of type 1 users increases. For bundled pricing, however, the minimum utility increases when either revenue or fairness is maximized, likely because type 1 users, who had previously received lower utilities than type 2 or 3 users, could now process more jobs. When fewer type 1 users are present, the price of type 1 users' jobs under resource or differentiated pricing may be set so that these users receive slightly higher utility values; since there are relatively few type 1 users, they do not take away many resources from the other types of users.

% We next change users' resource requirements so that they are perfectly symmetric...

% then suppose that users have the same number of What happens when we alter the distribution of users? Do the same qualitative results hold? What happens if we consider a strictly symmetric scenario? ({\color{red} maybe just show the fairness-revenue tradeoff})

\subsection{Volume Discounts}\label{sec:discounts}

We conclude by examining a parameter set by the operator: the volume discount $\gamma$. Figure \ref{fig:gamma_fair_rev} shows the fairness-revenue tradeoffs for a range of volume discounts $\gamma$. 

We see that for both bundled and resource pricing, varying $\nu$ in the operator's objective function (\ref{eq:opt}) does not yield an appreciable tradeoff (for bundled pricing, this result is again Corollary \ref{cor:bundle}). Moreover, as $\gamma$ decreases, the fairness decreases quite rapidly under both these pricing plans; once the revenue exceeds 50, the fairness value decreases exponentially. Differentiated pricing allows a wider range of much higher fairness values. As the revenue weight $\nu$ increases from 0, the revenue under differentiated pricing increases rapidly, while the range of achieved fairness stays roughly the same; though each contour represents a linear increase in $\nu$, the resulting increase in revenue, measured by the ``distance'' between the contours, is very pronounced. As $\nu$ increases further, the contours grow closer together, but the range of achieved fairness decreases more rapidly.

For larger revenue weights $\nu$, we find that the operator's revenue decreases with $\gamma$ for all three pricing plans, as a result of decreased demand from users. This demand decrease, however, affects users in different ways and acts to increase fairness. When $\nu = 0$ under differentiated pricing, the revenue does not decrease monotonically with $\gamma$; for a small volume discount ($\gamma$ near 1), the revenue initially increases and then decreases with $\gamma$. Unlike our results for varied memory capacities or distributions of users (Figs. \ref{fig:capacity_fair_rev} and \ref{fig:type_fair_rev}), the differentiated pricing contours for different $\nu$ actually overlap: the operator can achieve the same combination of fairness and revenue by optimizing two different weighted sums of fairness and revenue while offering different volume discounts.

We can explain the very negative fairness values observed in Figs. \ref{fig:gamma_bund_fair_rev} and \ref{fig:gamma_res_fair_rev} for bundled and resource pricing by examining the equitability and efficiency components of the fairness function (\ref{eq:betafairness}). We see that the equitability is quite negative, even for differentiated pricing, where the fairness values are relatively moderate (Fig. \ref{fig:gamma_diff_fair_rev}). Thus, the very low fairness values observed in Fig. \ref{fig:gamma_fair_rev} may be attributed more to lack of equitability in the utilities received by different users rather than a low total utility.

The relatively low efficiency values observed in Fig. \ref{fig:gamma_equ_eff} lead us to consider another measure of efficiency: leftover resource capacity. Though leftover capacity is not explicitly accounted for in the user's optimization problem (\ref{eq:opt})--capacity only enters through the resource constraints--the leftover capacity may indicate the potential additional revenue that an operator could achieve, e.g,. if users' demand functions change due to different volume discounts offered. We find that for all volume discounts $\gamma$ and weights $\nu$, only CPU cores are leftover; as the revenue increases, the leftover capacity increases for both resource and differentiated pricing. Thus, a volume discount reduces the ``wasted'' capacity by increasing user demand and operator revenue. For bundled pricing, however, we have the opposite scenario: as the revenue and overall user demand increase, the leftover capacity also increases. With bundled pricing, users are forced to purchase excess capacity according to the ratio of resources in the resource bundle. As the volume discount increases, user demand goes up and users purchase more excess capacity, which is included in the leftover capacity shown in Fig. \ref{fig:gamma_equ_eff}. This result thus highlights the ``inefficiency'' inherent in forcing users to pay for resources they do not need.

\begin{figure}
\centering
\includegraphics[width = 0.37\textwidth]{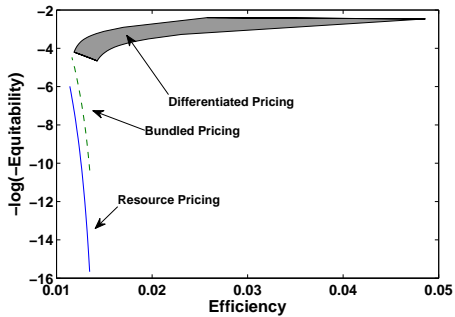}
	\vspace{-0.05in}
\caption{Equitability and efficiency components of the fairness function (\ref{eq:betafairness}) as the revenue weight $\nu$ varies in the operator's objective function (\ref{eq:opt}).}
\label{fig:gamma_equ_eff}
	\vspace{-0.15in}
\end{figure}

\begin{figure}
\centering
\includegraphics[width = 0.42\textwidth]{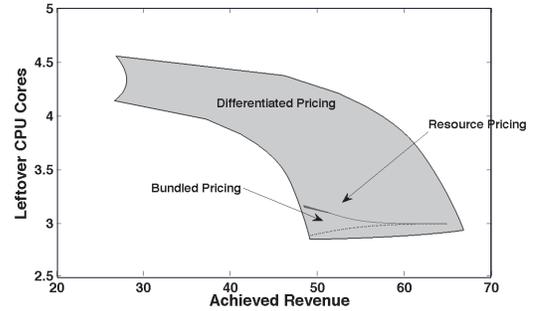}
	\vspace{-0.05in}
\caption{Leftover CPU capacity versus achieved revenue as the volume discount $\gamma$ and revenue weight $\nu$ vary in the operator's objective function (\ref{eq:opt}).}
\label{fig:gamma_cap_rev}
\vspace{-0.15in}
\end{figure}

% Finally, we vary the memory capacity as in Section \ref{sec:capacity} and measure the leftover capacity. We find that as when the volume discount varies in Fig. \ref{fig:gamma_cap_rev}, the leftover capacity decreases with achieved revenue for differentiated pricing (Fig. \ref{fig:capacity_cap_rev}.

%\begin{figure}
%\centering
%\includegraphics[width = 0.35\textwidth]{Matlab/Figures/Capacity_Capacity_Revenue.eps}
%\caption{Leftover CPU capacity versus achieved revenue as the memory capacity and revenue weight $\nu$ vary in the operator's objective function (\ref{eq:opt}).}
%\label{fig:capacity_cap_rev}
%\end{figure}

\section{Conclusion and Future Work}\label{sec:conclude}

In this work, we develop a mathematical framework for evaluating three cloud pricing strategies--bundled, resource, and differentiated pricing--in terms of their resulting fairness and revenue. We first characterize user demand for resources as a function of the prices offered under these different pricing plans, and then find the resulting revenue. We next formulate a family of metrics to measure the fairness of the offered prices and analyze the tradeoff between fairness and revenue. After showing some analytical bounds on this tradeoff, we finally use data taken from a Google cluster to numerically evaluate the impact of resource capacity and volume discounts on the operator's fairness-revenue tradeoff. We find that differentiated pricing, even if it only adds one additional price variable compared to resource pricing, can significantly increase the fairness under bundled or resource pricing, though revenue remains comparable across all pricing strategies. The operator can also increase both fairness and revenue through an increase in resource capacity, but the magnitude of this benefit depends heavily on the heterogeneity of users' resource requirements. Finally, a lower volume discount can be used to increase fairness, but at a loss of some revenue and increase in unused resource capacity.

Throughout this work, we consider fairness and revenue solely in terms of prices. A real implementation, however, would allow us to explore a different aspect of fairness: user priorities. Instead of choosing prices so as to balance fairness and revenue, an operator could instead leverage the temporal dimension of job processing in datacenters to orthognalize an operator's desire for fairness with its incentive to increase revenue. Users receiving lower utility values could be compensated with higher priorities, allowing their jobs to be completed faster. By implementing such a prioritized allocation scheme and devising algorithms for computing these priorities, we can evaluate whether setting user priorities, choosing job prices under one of many possible pricing plans, or a combination of both is most effective at balancing a datacenter operator's competing objectives of maximizing revenue and fairness.

\bibliographystyle{abbrv}
\bibliography{References}

\newpage

\appendix
% \section*{APPENDIX}

\section{Introducing Deadlines}\label{sec:deadlines}

In this section, we show that our formulation in Sections \ref{sec:user} - \ref{sec:fair_rev} may be altered to allow changes in user utilities and job demands over multiple time intervals. In formulating this problem, we account for varying job deadlines: users specify a completion deadline when a job is submitted, and the operator takes these into account when setting prices for each time interval.  We suppose that the operator sets different prices in each time interval; then if the price optimization for the single interval is convex (e.g., as in Prop. \ref{prop:concave}), so is the optimization over prices in all time intervals. We derive this extension only for resource pricing, but similar methods can be used for bundled or differentiated pricing.

We suppose that when users submit jobs to the datacenter operator in each time interval, they specify the maximum time to complete these jobs as well as the job resource requirements that, if met, will ensure that the job is completed before the deadline. We suppose that it does not matter when a user receives the requested resources, so long as the resources are allocated in the requested ratio. For instance, a job with a deadline of three time intervals and a requirement of 6 CPU slots and 3 GB of RAM might be allocated 2 slots and 1 GB in each time interval; or 2 slots and 1 GB in the first, no resources in the second, and 4 slots and 2 GB in the third time interval; either allocation allows the job to complete before the deadline. In each time interval, each user submits one type of job to the datacenter (i.e., $R_{ij}$ is the same for all jobs of user $j$ and all resources $i$ submitted during a given time interval). If a user has multiple types of jobs to run, these may be treated as coming from distinct users. One could also group users with similar per-job resource requirements.

We consider a finite time horizon of $T$ time intervals, indexed by $s = 1,2,\ldots, T$, and assume that all job deadlines fall before the final time $T$. At the beginning of each time interval $s$, the operator sets the resource prices $p_{i,s}$ for each resource $i$; denoting the vector of resource prices in period $s$ by ${\bf p_s}$, each user $j$ then submits $x^\star_{j,s}({\bf p_s})$ number of jobs to the operator, with associated deadline $\tau_{j,s}$. These jobs are charged according to the price ${\bf p_s}$, though the resource prices may change before the job completes. Users choose the number of jobs so as to maximize their utility in period $s$, and do not account for any jobs that have been already submitted or may be submitted in future periods. For simplicity, we suppose that for each user $j$, all of user $j$'s jobs submitted in a given period $s$ have the same deadline.

We now introduce the variables $x_{j,s}(t)$, $s < t$, to denote the number of jobs submitted by user $j$ in period $s$ that are processed during time interval $t$. To ensure that all deadlines are met, we introduce the constraints
\begin{equation}
\sum_{t = s}^{\tau_{j,s}} x_{j,s}(t) \geq x^\star_{j,s}({\bf p_s})
\label{eq:deadline}
\end{equation}
for each user $j$ and time interval $s$; since $x^\star_{j,s}$ is convex in the prices ${\bf p_s}$ by Lemma \ref{lem:resourceconvex}, this constraint defines a convex set. The resource constraints for each period $t$ then become
\begin{equation}
\sum_{s = 1}^t\sum_{j = 1}^n R_{ij}(s)x_{j,s}(t) \leq C_i,
\label{eq:resource_deadline}
\end{equation}
where $R_{ij}(s)$ denotes the per-job requirement of resource $i$ for jobs submitted by user $j$ in the time interval $s$. We note that these constraints are affine in the variables $x_{j,s}(t)$; thus, the constraints (\ref{eq:deadline}) and (\ref{eq:resource_deadline}) together define a convex set in the variables ${\bf p_s}$ and $x_{j,s}(t)$. We then find that the revenue with deadlines may be expressed as
\begin{equation}
\sum_{s = 1}^T \left(\sum_{j = 1}^n {x_{j,s}^\star}^\gamma\sum_{i = 1}^m R_{ij}(s)p_{i,s}\right).
\label{eq:revenue_deadlines}
\end{equation}
The fairness may be thought of as the sum of the fairness of the utilities received in each time interval; the fairness function is thus only indirectly dependent on the deadlines $\tau_{j,s}$. Users specify the deadlines independent of the prices offered, so these should not affect their perceived utility or fairness. Thus, (\ref{eq:betafairness}) yields the fairness function
\begin{equation}
\frac{1}{1 - \beta}\sum_{s = 1}^T \left(\sum_{j = 1}^n \overline{U}_j\left(x_{j,s}^\star({\bf p_s}), {\bf p_s}\right)^{1 - \beta}\right).
\label{eq:fair_deadlines}
\end{equation}
We note that since both the revenue (\ref{eq:revenue_deadlines}) and fairness (\ref{eq:fair_deadlines}) are simply the revenue and fairness in each time interval $s$ added together, we may apply Prop. \ref{prop:concave} to formulate a weighted sum (with weights possibly depending on the time interval) of both the revenue and fairness that is a concave function. We then have the following result:
\begin{thm}
The optimization problem
\begin{align}
\max_{x_{j,s}(t), p_{i,s}}&\;\sum_{s = 1}^T \nu(s)\left(\sum_{j = 1}^n {x_{j,s}^\star}^\gamma\sum_{i = 1}^m R_{ij}(s)p_{i,s}\right) \nonumber \\
&+ \frac{1}{1 - \beta}\left(\sum_{j = 1}^n \overline{U}_j\left(x_{j,s}^\star({\bf p_s}), {\bf p_s}\right)^{1 - \beta}\right) \\
{\rm s.t.}&\;\sum_{s = 1}^t\sum_{j = 1}^n R_{ij}(s)x_{j,s}(t) \leq C_i;\;1\leq t\leq T;\;1\leq i \leq m \\
&\sum_{t = s}^{\tau_{j,s}} x_{j,s}(t) \geq x^\star_{j,s}({\bf p_s});\;1 \leq s \leq T;\; 1\leq j \leq n
\end{align}
is a convex optimization if, for each time interval $s$, the weight $\nu(s)$ is chosen so as to satisfy Prop. \ref{prop:concave} for the revenue and fairness in period $s$.
\end{thm}

\section{Proofs}\label{sec:proofs}

\subsection{Proposition \ref{prop:concave}}
\begin{proof}
As in the proof of Lemma \ref{lem:resourceconvex}, we show the result for resource pricing; the proofs for bundled or differentiated pricing are analogous. We first note that the second derivative matrix of (\ref{eq:weighted}) may be written as
\begin{equation*}
\nu\nabla^2\rho + {\bf \nabla^2 F_\beta} = {\bf \left[R_{ij}^\gamma\right] Q\left[R_{ij}^\gamma\right]}^T,
\end{equation*}
where ${\bf Q}$ is a diagonal matrix with entries
\begin{align*}
Q_{jj} &= \nu\left(\gamma r_j({\bf p})\right)^{\frac{2\gamma + \alpha_j - 1}{1 - \alpha_j - \gamma}}\gamma^2\frac{1 - \alpha_j}{\left(1 - \alpha_j - \gamma\right)^2} \\
&+ \gamma^{\frac{(1 - \beta)\gamma}{1 - \alpha_j - \gamma}}\left(\frac{\gamma + \alpha_j - 1}{1 - \alpha_j}\right)^{-\beta}\frac{\beta(1 - \alpha_j) - \gamma}{1 - \alpha_j - \gamma} \\
&\times r_j^{\frac{(-1-\beta)(1 - \alpha_j) + 2\gamma}{1 - \alpha_j - \gamma}}
\end{align*}
and ${\bf \left[R_{ij}^\gamma\right]}$ is an $m\times n$ matrix with $(i,j)$ entry $R_{ij}^\gamma$. A sufficient condition for this second-derivative matrix to be negative-semidefinite is that $Q_{jj} < 0$ for each $j$. Thus, we find the condition
\begin{align*}
\nu \leq &\left(\frac{\gamma + \alpha_j - 1}{1 - \alpha_j}\right)^{1 - \beta}\gamma^{\frac{\beta\gamma}{\alpha_j + \gamma - 1} - 1}\big(\beta(1 - \alpha_j) - \gamma\big) \\
&\times \left(r_j\right)^{\frac{\beta(1 - \alpha_j)}{\alpha_j + \gamma - 1}}.
\end{align*}.
We can remove the price variables from this bound by observing that the resource constraints ${\bf Rx^\star} \leq {\bf C}$ imply that
\begin{equation*}
R_{ij}r_j^{\frac{\gamma}{1 - \alpha_j - \gamma}} \leq C_i
\end{equation*}
for each resource $i$ and user $j$; thus,
\begin{equation*}
r_j^{\frac{\beta(1 - \alpha_j)}{\alpha_j + \gamma - 1}} \geq \max_i\left(\frac{C_i}{R_{ij}}\right)^{\frac{\beta(\alpha_j - 1)}{\gamma}}.
\end{equation*}
We therefore find the sufficient condition
\begin{align*}
\nu \leq &\left(\frac{\gamma + \alpha_j - 1}{1 - \alpha_j}\right)^{1 - \beta}\gamma^{\frac{\beta\gamma}{\alpha_j + \gamma - 1} - 1}\big(\beta(1 - \alpha_j) - \gamma\big) \\
&\times \left(\max_i\left(\frac{C_i}{R_{ij}}\right)^{\frac{\beta(\alpha_j - 1)}{\gamma}}\right).
\end{align*}
\end{proof}

\subsection{Proposition \ref{prop:bound}}
\begin{proof}
Suppose that $\beta < 1$ in (\ref{eq:betafairness}). We first introduce the notation
\begin{equation*}
\nu_j = \frac{\gamma + \alpha_j - 1}{1 - \alpha_j} = \frac{\gamma}{1 - \alpha_j} - 1,
\end{equation*}
and find that the fairness may be written as
\begin{align*}
&\frac{1}{1 - \beta}\sum_{j = 1}^n \nu_j^{1 - \beta}\left(\frac{\gamma}{c_j}\right)^{\frac{\gamma(1 - \beta)}{1 - \alpha_j - \gamma}}r_j^{\frac{(1 - \alpha_j)(1 - \beta)}{1 - \alpha_j - \gamma}} \\
&\geq \frac{\min_j \nu_j^{1 - \beta}}{1 - \beta}\left(\sum_{j = 1}^n \left(\frac{\gamma}{c_j}\right)^{\frac{\gamma}{1 - \alpha_j - \gamma}}r_j^{\frac{1 - \alpha_j}{1 - \alpha_j - \gamma}}\right)^{1 - \beta} \\
&= \frac{\rho\left({\bf p}\right)^{1 - \beta}}{1 - \beta}\left(\frac{\gamma}{1 - \alpha_k} - 1\right)^{1 - \beta}.
\end{align*}
The first inequality uses the sub-additivity of the function $f(x) = x^{1 - \beta}$, and we let $k$ be such that $\alpha_k = \min_j \alpha_j$.

We now suppose that $\beta > 1$ and find that for each user $j$,
\begin{equation*}
\nu_j^{1 - \beta}\left(\frac{\gamma}{c_j}\right)^{\frac{\gamma(1 - \beta)}{1 - \alpha_j - \gamma}}r_j^{\frac{(1 - \alpha_j)(1 - \beta)}{1 - \alpha_j - \gamma}} \leq F_\beta({\bf p})(1 - \beta),
\end{equation*}
which is equivalent to
\begin{equation*}
\left(\frac{\gamma}{c_j}\right)^{\frac{\gamma}{1 - \alpha_j - \gamma}}r_j^{\frac{1 - \alpha_j}{1 - \alpha_j - \gamma}} \geq \nu_j^{-1}\big(F_\beta({\bf p})(1 - \beta)\big)^{\frac{1}{1 - \beta}}.
\end{equation*}
Adding these lower bounds and using the definition of $\nu_j$, we find that
\begin{equation*}
\rho({\bf p}) \geq \big(F_\beta({\bf p})(1 - \beta)\big)^{\frac{1}{1 - \beta}}\sum_{j = 1}^n \frac{1 - \alpha_j}{\gamma + \alpha_j - 1}.
\end{equation*}
\end{proof}

\section{Fairness Properties}\label{sec:properties}

While the fairness functions (\ref{eq:betalambda}) are the unique functions satisfying certain desirable fairness axioms, other desirable fairness properties are not axiomatically satisfied by these functions. In this appendix, we consider two such properties: envy-freeness and Pareto-efficnecy.
\begin{defn}
{\bf Envy-freeness} holds if and only if no user envies another user's allocation.  Mathematically, let $r_{ij}$ denote the amount of resource $i$ received by user $j$.  User $j$ can then process $\max_i r_{ij}/R_{ij}$ number of jobs.  Envy-freeness is def\hspace{0.1mm}ined as the property that $\max_i r_{ij}/R_{ij} > \max_i r_{ik}/R_{ij}$ for any $j\neq k$. In words, no other user's allocation would enable a user to process more jobs than her allocation would.
\end{defn}
\begin{defn}
A function $f$ is {\bf Pareto-eff\hspace{0.1mm}icient} if, whenever ${\mathbf x}$ Pareto-dominates ${\mathbf y}$ (i.e., $x_i \geq y_i$ for each index $i$ and $x_j > y_j$ for some $j$), $f({\mathbf x}) > f({\mathbf y})$.
\end{defn}

We first show that envy-freeness always holds for an allocation of utilities if users optimize the number of jobs that they submit under bundled or resource pricing:
\begin{thm}
Suppose that each user chooses a number of jobs to submit, $x_j^\star(r_j\left({\bf p}\right))$, so as to maximize her utility function. Then for any given set of resource or bundled prices ${\bf p}$, if two users switch their corresponding resource allocations and are required to pay the difference in resources received, neither gains any utility. If users do not need to pay this difference, then a user may gain utility by switching allocations.
\end{thm}
\begin{proof}
The first part of the proposition is trivial, as for a given set of prices, users are assumed to optimize the number of jobs processed. Since all users are charged at the same per-resource or per-bundle rates, if switching resource allocations alters the number of jobs processed for either, neither gains any utility. If switching allocations results in each user being able to process the same number of jobs, then the users have either not changed their resource allocations or have gained resources; in the latter case, their utility decreases due to paying for extra resources that they do not use to process jobs.

If users are not required to pay the difference in prices from switching resource allocations, then one user may gain utility by switching resource allocations with another: consider, for instance, a user who processes a very small amount of jobs, due to a low willingness to pay (e.g., a high value of $\alpha_j$ and low value of $c_j$ in (\ref{eq:alphafair})'s utility functions). Then this user will process a higher number of jobs by switching resources allocations with another user who is more willing to pay for resources in exchange for processing jobs.
\end{proof}
We next find that the fairness function (\ref{eq:betalambda}) is Pareto-efficient if $\beta > 0$ and $\left|\lambda\right| \geq 1/\beta - 1$ \cite{joe2012multi,lan2010axiomatic}. Thus, if the allocation $\overline{U}_j\left(x_j^\star, {\bf p}\right)$ Pareto-dominates the allocation $\overline{U}_j\left(x_j^\star, {\bf q}\right)$, then the fairness function (\ref{eq:betafairness}), for which $\lambda = 1/\beta - 1$, attains a larger value under the prices ${\bf p}$ than it does under the prices ${\bf q}$. The prices here may represent those under any pricing plan; bundled, resource, or differentiated. If Lemma \ref{lem:const} holds and the utility received from each user $j$ is equal to a price-independent constant multiplied by the amount paid by user $j$, then the revenue (\ref{eq:revenue}) is also Pareto-efficient, as it is an affine transformation of the sum of the utilities received. This latter quantity corresponds to the fairness (\ref{eq:betalambda}) as $\lambda\rightarrow\infty$.

\end{document}